\documentclass[11pt,a4paper,oneside,reqno]{amsart}

\usepackage[latin1]{inputenc}
\usepackage{amsmath, amsthm, amssymb, amsfonts}
\usepackage[foot]{amsaddr}
\usepackage{algorithm2e}
\usepackage{bbding}
\usepackage{booktabs}
\usepackage{cite}
\usepackage[colorlinks=true, linkcolor = blue, citecolor = dgreen]{hyperref}
\usepackage{cleveref}
\usepackage{color}
\usepackage{diagbox}
\usepackage{dsfont}
\usepackage[english]{babel}
\usepackage{enumitem}
\usepackage{float}
\usepackage[a4paper]{geometry} 
\usepackage{graphicx}
\usepackage{mathrsfs}
\usepackage[all]{xy} 
\usepackage[numbers]{natbib}
\usepackage{rotating}
\usepackage{setspace}
\usepackage{stmaryrd}
\usepackage{subcaption}
\usepackage{xspace}
\usepackage{xcolor}
\usepackage{tikz}
\setlist{leftmargin=*, topsep=0.5em, parsep=0pt, itemsep=1em, labelindent=0pt, align=left}

\theoremstyle{definition}
\newtheorem{theorem}{Theorem}[section]

\newtheorem{proposition}{Proposition}[section]
\newtheorem{remark}{Remark}[section]
\newtheorem{assumption}{Assumption}[section]
\newtheorem{lemma}{Lemma}[section]
\newtheorem{corollary}{Corollary}[section]

\numberwithin{equation}{section}
\numberwithin{table}{section}
\newcommand{\littleo}[1]{o(#1)\xspace}
\newcommand{\littleoo}[1]{o_0(#1)\xspace}

\newcommand{\Qro}{\mathbb{Q}}
\newcommand{\Ex}{\mathbb{E}}
\newcommand{\Var}{\mathrm{Var}}
\newcommand{\Cov}{\mathrm{Cov}}

\newcommand{\EE}{\mathcal{E}}
\newcommand{\FF}{\mathcal{F}}

\newcommand{\NN}{\mathcal{N}}

\newcommand{\dd}{\mathrm {d}}

\newcommand{\rind}[1]{\textbf{1}_{#1}}

\newcommand{\reals}{\mathbb{R}}

\newcommand{\naturals}{\mathbb{N}}

\definecolor{dgreen}{rgb}{0,.8,0}
\definecolor{red}{HTML}{D62728}
\definecolor{blue}{RGB}{ 0, 109, 219}
\definecolor{dgreen}{rgb}{0,.8,0}

\crefformat{enumi}{item~(#2#1#3)}
\Crefformat{enumi}{Item~(#2#1#3)}
\crefrangeformat{enumi}{items~(#3#1#4) to~(#5#2#6)}
\Crefrangeformat{enumi}{Items~(#3#1#4) to~(#5#2#6)}

\newcommand\blfootnote[1]{%
  \begingroup
  \renewcommand\thefootnote{}\footnote{#1}%
  \addtocounter{footnote}{-1}%
  \endgroup
}

\begin{document}
\title[]{Closed-form approximations\\ with respect to the mixing solution \\for option pricing under stochastic volatility}
\author{Kaustav Das$^\dagger$}
\address{$^\dagger$School of Mathematics, Monash University, Victoria, 3800 Australia.}
\author{Nicolas Langren\'e$^\ddagger$}
\address{$^\ddagger$CSIRO Data61, Victoria, 3008 Australia.}
\email{kaustav.das@monash.edu, nicolas.langrene@csiro.au}
\date{}
%\titleformat*{\section}{\large \bfseries}
%\titleformat*{\subsection}{\normalsize\bfseries}
\DeclareGraphicsExtensions{.pdf,jpg,png}
\vspace{12pt}

\onehalfspacing

%---------------------------------------------------------------------------------Abstract------------------------------------------------------------------------------------
%--------------------------------------------------------------------------------------------------------------------------------------------------------------------------------
%--------------------------------------------------------------------------------------------------------------------------------------------------------------------------------
\begin{abstract}
\noindent We consider closed-form approximations for European put option prices within the Heston and GARCH diffusion stochastic volatility models with time-dependent parameters. Our methodology involves writing the put option price as an expectation of a Black-Scholes formula and performing a second-order Taylor expansion around the mean of its argument. The difficulties then faced are simplifying a number of expectations induced by the Taylor expansion. Under the assumption of piecewise-constant parameters, we derive closed-form pricing formulas and devise a fast calibration scheme. Furthermore, we perform a numerical error and sensitivity analysis to investigate the quality of our approximation and show that the errors are well within the acceptable range for application purposes. Lastly, we derive bounds on the remainder term generated by the Taylor expansion. \\[.5cm]
\end{abstract}
\maketitle
\blfootnote{Research supported by an Australian Government Research Training Program (RTP) Scholarship.}
%---------------------------------------------------------------------------------Introduction--------------------------------------------------------------------------------
%---------------------------------------------------------------------------------------------------------------------------------------------------------------------------------
\section{Introduction}
%---Our contribution/big picture
\noindent We consider the European put option pricing problem within the Heston and GARCH\footnote{Generalised AutoRegressive Conditional Heteroskedasticity.} diffusion stochastic volatility models, see \citep{heston} and \citep{nelson1990arch, Willard97} respectively. Our goal is to study how European put option prices can be approximated in each of these frameworks via expansion of the so-called mixing solution, which will be detailed later in this article. We find that through a second-order Taylor expansion of the mixing solution, we can derive accurate approximations to European put option prices in the aforementioned models. In addition, our methodology works naturally with time-dependent parameters. This is seen as a major advantage compared to other methodologies, for example transform methods, which cannot handle time-dependent parameters well. Furthermore, the approximation formulas are written solely in terms of iterated integrals, which we show obey convenient recursive properties when parameters are assumed to be piecewise-constant. This results in the approximation being essentially closed-form, and moreover leads to a fast calibration scheme. Our method builds upon \citep{drimus2011closed}, in which the Heston model is considered with constant parameters. A sensitivity analysis is performed in order to assess our approximation numerically. Furthermore, we present mathematical bounds on the error term in terms of higher order moments of the underlying variance process.

%---Motivation for closed-form pricing, i.e., calibration
It has been well established that volatility is highly dependent on the strike and maturity of European option contracts. This phenomenon is called the volatility smile or skew, an attribute the well known Black-Scholes model \citep{bs} fails to take into account, see for example \citep{gatheral2011volatility} for more on this. In response, there have been a number of frameworks proposed to explain the volatility smiles and skews observed in the market, for example, local volatility models, stochastic volatility models, and local-stochastic volatility models. In particular, stochastic volatility models have been introduced, where the volatility itself is a stochastic process possibly correlated with the spot. However with this added complexity, often option prices cannot be computed in a closed-form fashion. This is detrimental, as closed-form solutions lead to rapid option pricing, an attribute necessary for fast calibration of financial models. Without a closed-form solution, option pricing must be done numerically via Monte-Carlo or PDE methods, both methods being computationally costly.

%---Affine models
If we assume that the characteristic function of the log-spot is known explicitly, then the option price can be computed quasi-explicitly (meaning in terms of at most one-dimensional complex integrals, where the integrands are explicit function), albeit under the restrictive assumption of constant or piecewise-constant parameters \citep{heston, carr1999option, mikhailov2004heston}. One class of models where this occurs is the class of affine models such as the Heston and Sch{\"o}bel-Zhu model. However, for non-affine models, the characteristic function of the log-spot is rarely known explicitly, and such a procedure will not be effective. Non-affine models, although usually intractable compared to their affine counterparts, are often far more realistic. This has been shown in a number of studies, see for example \citep{gander2007stochastic,christoffersen2010volatility,kaeck2012volatility}. For these reasons, numerical procedures such as PDE and Monte-Carlo methods have been substantially developed in the literature \citep{van2014heston,andersen2007efficient,kahl2006fast}.

%---Closed-form approx
Closed-form approximations are an alternative methodology for option pricing, where the option price is approximated by a closed-form expression. The main advantages are that the option price can be computed rapidly and since transform methods are not used, time-dependent parameters can usually be handled well. One motivation for quick option pricing formulas is calibration, where the option price must be computed several times within an optimisation procedure.
 
%---Closed-form approx lit review
There have been a plethora of results on closed-form expansions in the literature. For example, \citep{lorigexplicit} derive a general closed-form expression for the price of an option via a PDE approach, as well as its corresponding implied volatility. \citep{sabr} use singular perturbation techniques to obtain an explicit approximation for the option price and implied volatility in their SABR model. \citep{alos2006generalization} show that from the mixing solution, one can approximate the put option price by decomposing it into a sum of two terms, one being completely correlation independent and the other dependent on correlation. However, neither terms are explicit. Furthermore, similar to our work, \citep{Antonelli09,Antonelli10} show that under the assumption of small correlation, an expansion can be performed with respect to the mixing solution, where the resulting expectations can be computed using Malliavin calculus techniques. Similarly, in the case of the time-dependent Heston model, \citep{timedepheston} consider the mixing solution and expand around vol-of-vol, performing a combination of Taylor expansions and computing the resulting terms via Malliavin calculus techniques.

%---The method we use. Drimus limitation
Stochastic volatility models usually either model the volatility directly, or indirectly via the variance process. A critical assumption is that volatility or variance has some sort of mean reversion behaviour, and this is supported by empirical evidence, see for example \citep{gatheral2011volatility}. Specifically, for modelling the variance, a large class of stochastic volatility models is given by
\begin{align*}
	\dd S_t &= S_t ( (r_t^d - r_t^f ) \dd t + \sqrt{V_t} \dd W_t ), \quad S_0,\\
	\dd V_t &=\kappa_t ( \theta_t V_t^{\hat \mu} - V_t^{\tilde \mu }) \dd t + \lambda_t V_t^{\mu} \dd B_t , \quad V_0 = v_0,\\
	\dd \langle W, B \rangle_t &= \rho_t \dd t,
\end{align*}
whereas for modelling the volatility, this class is of the form
\begin{align*}
	\dd S_t &= S_t ( (r_t^d - r_t^f ) \dd t + V_t \dd W_t ), \quad S_0,\\
	\dd V_t &=\kappa_t ( \theta_t V_t^{\hat \mu} - V_t^{\tilde \mu }) \dd t + \lambda_t V_t^{\mu} \dd B_t , \quad V_0 = v_0,\\
	\dd \langle W, B \rangle_t &= \rho_t \dd t,
\end{align*}	
for some $ \tilde \mu, \hat \mu$ and  $\mu \in \reals$.\footnote{There exist other classes of stochastic volatility models, for example the exponential Ornstein-Uhlenbeck model \citep{wiggins1987option} is not a part of either of these classes.} Our model formulation here is for FX market purposes, but can be easily adjusted for equity and fixed income markets. In this article, we will focus on this class of stochastic volatility models. Some popular models in the literature include: \\[1em]
\noindent
\begin{tabular}{| l | c | l | l | l | l | }
\hline
Model &  Variance/Volatility  &Dynamics of $V$ & $\hat \mu$ & $\tilde \mu$ & $\mu$  \\
\hline
Heston \text{\citep{heston}} & Variance & $\dd V_t = \kappa_t( \theta_t - V_t) \dd t +  \lambda_t \sqrt {V_t} \dd B_t $ & 0  &1 &  1/2 \\
\hline
Sch{\"o}bel/Zhu \text{\citep{schobel1999stochastic}} & Volatility &$ \dd V_t = \kappa_t( \theta_t - V_t) \dd t + \lambda_t \dd B_t$ & 0 & 1 & 0 \\
\hline
GARCH \text{\citep{nelson1990arch, Willard97}} & Variance  & $\dd V_t = \kappa_t( \theta_t - V_t) \dd t +  \lambda_t V_t \dd B_t$  & 0 & 1 & 1\\
\hline
Inverse-Gamma \text{\citep{IGa}} & Volatility & $\dd V_t = \kappa_t( \theta_t - V_t) \dd t +  \lambda_t V_t \dd B_t $ & 0 & 1 & 1  \\
\hline
3/2 Model \text{\citep{lien2002option}}  & Variance & $\dd V_t = \kappa_t( \theta_t V_t - V_t^2) \dd t +  \lambda_t V_t^{3/2} \dd B_t$  & 1 & 2 & 3/2\\
\hline
Verhulst \text{\citep{lewis2019exact,carr2019lognormal}} & Volatility &  $\dd V_t = \kappa_t( \theta_t V_t - V_t^2) \dd t +  \lambda_t V_t \dd B_t $ & 1 & 2 & 1  \\
\hline
\end{tabular}

\noindent Note that the Verhulst model also goes by the names Logistic model and XGBM model.
%\clearpage

%--Extend Drimus. Detail on contribution
This article is dedicated to detailing how a second-order expansion of the so-called mixing solution for stochastic volatility models with time-dependent parameters can result in an approximation for the price of a European put option. This approximation is always in terms of some specific iterated integrals. We then show that these iterated integrals obey a convenient recursive scheme. Not only does this lead to the approximation being essentially closed-form, but this also results in a fast calibration scheme. The tractability of our methodology relies largely on the dynamics of the underlying variance process. Our methodology extends that of \citep{drimus2011closed}, in which the Heston model is considered with constant parameters. Our contribution is the following. We apply the expansion methodology in the Heston and GARCH diffusion stochastic volatility models with time-dependent parameters. In particular, the GARCH diffusion model is an example of a non-affine model, which practitioners prefer, see for example \citep{gander2007stochastic,christoffersen2010volatility,kaeck2012volatility}. We include a robust error analysis, design a fast calibration scheme and give extensive numerical results. The sections are organised as follows: 
\begin{itemize}
	\item \Cref{sec:prelims2} details preliminary calculations, where we express the put option price as the mixing solution. Once done, a second-order Taylor expansion is performed, expressing the approximation formula in terms of a number of expectations of functionals of the underlying variance process.
	\item \Cref{sec:mom} details how to derive more convenient expressions for these expectations obtained in \Cref{sec:prelims2} through change of measure techniques. Specifically, we rewrite the spot $S_T$ as a convenient expression so as to construct a term which is a Dol\'eans-Dade exponential, thereby defining a Radon-Nikodym derivative. This term allows us to change measure, allowing for more convenient calculations.
	\item \Cref{sec:pricingspecmodels2} introduces specific models. As precise dynamics are assumed, the objective is to derive expressions in terms of specific iterated integrals for the expectations from \Cref{sec:mom}. In particular, we consider the Heston and GARCH diffusion models.
	\item In \Cref{sec:erroranalysis2} we perform an error analysis, bounding the error in the expansion in terms of higher order moments of the underlying variance process.
	\item In \Cref{sec:fastcal2}, under the assumption of piecewise-constant parameters, we obtain closed-form expressions for our pricing formulas. Moreover, we design a fast calibration scheme. Specifically, we show that the iterated integrals in which the pricing formulas in \Cref{sec:pricingspecmodels2} are in terms of satisfy some convenient recursive properties.
	\item \Cref{sec:numerical2} is dedicated to a numerical error and sensitivity analysis for the Heston and GARCH diffusion models.
\end{itemize} 

%----------------------------------------------------------------------------Preliminary calculations---------------------------------------------------------------------
%----------------------------------------------------------------------------------------------------------------------------------------------------------------------------
\section{Preliminaries}
\noindent
\label{sec:prelims2}
%----Dynamics 
Suppose the spot  $S$ with variance $\sigma$ follows the dynamics
\begin{align*}
	\dd S_t &= S_t ( (r_t^d - r_t^f ) \dd t + \sqrt{\sigma_t} \dd W_t ), \quad S_0,\\
	\dd \sigma_t &= \alpha(t, \sigma_t) \dd t + \beta(t, \sigma_t) \dd B_t , \quad \sigma_0,\\
	\dd \langle W, B \rangle_t &= \rho_t \dd t,
\end{align*}
where $W$ and $B$ are Brownian motions with deterministic, time-dependent instantaneous correlation $(\rho_t)_{0 \leq t \leq T}$, defined on the filtered probability space $(\Omega, \FF, (\FF_t)_{0 \leq t \leq T}, \Qro)$. Here $T$ is a finite time horizon, where $(r_t^d)_{0\leq t \leq T}$ and $(r_t^f)_{0 \leq t \leq T}$ are the deterministic, time-dependent domestic and foreign interest rates respectively. In addition, $\rho_t \in [-1,1]$ for any $t \in [0,T]$. Furthermore, $(\FF_t)_{0 \leq t \leq T}$ is the filtration generated by $(W,B)$ which satisfies the usual assumptions. In the following, $\Ex (\cdot)$ denotes the expectation under $\Qro$, where $\Qro$ is a domestic risk-neutral measure which we assume to be chosen. 

%------Assumption: variance process solution
\begin{assumption}
\label{ass:solutionvariance}
\noindent
 $\sigma$ possesses a non-negative pathwise unique strong solution which does not blow up in finite time.
\end{assumption}

We will enforce \Cref{ass:solutionvariance} for the rest of this article. We briefly remark that sufficient conditions which ensure pathwise uniqueness of the solution are the usual It\^o conditions (that is, Lipschitz continuity on the drift and diffusion coefficients) or the Yamada-Watanabe condition (Theorem 1 in \citep{yamada1971uniqueness}). Moreover, in order to ensure that the solution does not blow up in finite time, linear growth conditions on the drift and diffusion coefficients are sufficient.

%-------Remark: GBM process
\begin{remark}[Geometric Brownian motion process]
Let $\tilde B$ be an arbitrary Brownian motion process. We call $Y$ a GBM$(y_0 ; \mu_t, \nu_t)$ if it solves the SDE 
\begin{align*}
\dd Y_t = \mu_t Y_t \dd t+ \nu_t Y_t \dd \tilde B_t, \quad Y_0 = y_0,
\end{align*} 
where $\mu$ and $\nu$ are adapted to the Brownian filtration and satisfy some regularity conditions (for example, $\mu$ and $\nu$ bounded on $[0, T]$ is sufficient).
\end{remark}
We decompose the Brownian motion $W$ as $W_t = \int_0^t \rho_u  \dd B_u+\int_0^t \sqrt{1 - \rho_u^2} \dd Z_u$, where $Z$ is a Brownian motion under $\Qro$ such that $B$ and $Z$ are independent. Then, noticing $S$ is a GBM$(S_0; r_t^d - r_t^f, \sqrt{\sigma_t})$, we obtain the expression
\begin{align*}
	S_T &= S_0 \xi_T \exp \left \{   \int_0^T (r_t^d - r_t^f ) \dd t - \frac{1}{2} \int_0^T \sigma_t  (1 - \rho_t^2) \dd t + \int_0^T  \sqrt{ \sigma_t (1 - \rho_t^2)} \dd 		Z_t \right \},
\end{align*}
where
\begin{align*}
	\xi_t &:= \exp \left \{ \int_0^t \rho_u \sqrt{\sigma_u}  \dd B_u - \frac{1}{2} \int_0^t \rho_u^2 \sigma_u \dd u \right \}. 
\end{align*}

%------Assumption: xi martingale
\begin{assumption}
\label{ass:martingale}
\noindent
$\sup_{t \in [0,T]} \limsup_{x \to \infty} \frac{ \rho_t \beta(t,x) \sqrt{x} + \alpha(t, x) }{x} < \infty.$ 
\end{assumption}
We will enforce \Cref{ass:martingale} for the rest of this article.

%------Lemma: Conditions on xi being martingale.
\begin{lemma}
\label{lemma:martingalexi}
$\xi$ is a martingale.
\end{lemma}
%------Lemma: Conditions on xi being martingale Proof
\begin{proof}[Proof]
The proof is inspired by arguments made in Theorem 2.4 (i) of \citep{lions2007correlations}. First notice that
\begin{align}
	\dd \xi_t &= \rho_t \sqrt{\sigma_t} \xi_t \dd B_t, \quad \xi_0 = 1. \label{sdexi}
\end{align}
Let $\tau_n := \inf \{t \geq 0: \sigma_t > n \}$ be the first time $\sigma$ crosses the level $n$. Then $\tau_n \uparrow \infty$ $\Qro$ a.s. Define $\xi_t^n := \xi_{t \wedge \tau_n}$. Then by solving \cref{sdexi} up to $\tau_n$, this yields the pathwise unique strong solution
\begin{align*}
	\xi_t^n =\exp \left (\int_0^t \rho_u \sqrt{\sigma_u} \rind{\{u \leq \tau_n\} } \dd B_u - \frac{1}{2}\int_0^t \rho_u^2 \sigma_u \rind{\{u \leq \tau_n\} } \dd u \right ).
\end{align*}
Since the integrand of the above It\^o integral is bounded, $(\xi_t^n)_t$ is a martingale for each $n \in \naturals$. Furthermore, $\xi$ is a non-negative local-martingale. Hence by Fatou's lemma, we have that $\xi$ is a non-negative supermartingale. Our goal now is to show that
\begin{align*}
	\sup_n \Ex( \xi_T^n \log(\xi_T^n) ) < \infty
\end{align*}
since then, by the Vall\'ee Poussin theorem, $(\xi_T^n)_n$ is uniformly integrable and thus ${\Ex(\xi_T) = 1}$. Combined with the fact that $\xi$ is a non-negative supermartingale, this will then imply that $\xi$ is a martingale. Now,
\begin{align*}
	\xi_T^n \log (\xi_T^n) = \xi_T^n \left ( \int_0^T \rho_t \sqrt{\sigma_t} \rind{\{t \leq \tau_n\} } \dd B_t - \frac{1}{2} \int_0^T \rho_t^2 \sigma_t  \rind{\{t \leq \tau_n\} } \dd t \right ).
\end{align*}
Clearly $\xi_T^n$ is a Radon-Nikodym derivative which defines a measure $\hat \Qro$. Hence by Girsanov's theorem, $\hat B_t = B_t - \int_0^t \rho_u \sqrt{\sigma_u} \rind{ \{ u \leq \tau_n \} }\dd u$ is a Brownian motion under $\hat \Qro$. Denote the expectation under $\hat \Qro$ by $\hat \Ex$. Then
\begin{align*}
	\Ex(\xi_T^n \log (\xi_T^n)) &= \Ex \left (\xi_T^n \left [\int_0^T \rho_t \sqrt{\sigma_t}  \rind{\{t \leq \tau_n\} } \dd B_t - \frac{1}{2} \int_0^T \rho_t^2 \sigma_t \rind{\{t \leq \tau_n\} } \dd t \right ]\right ) \\
	&=\hat \Ex \left (\int_0^T \rho_t \sqrt{\sigma_t} \rind{\{t \leq \tau_n\} } \dd \hat B_t + \frac{1}{2} \int_0^T \rho_t^2 \sigma_t\rind{\{t \leq \tau_n\} } \dd t \right ) \\
		&= \frac{1}{2} \int_0^T \hat \Ex( \rho_t^2 \sigma_t \rind{\{t \leq \tau_n\} } ) \dd t \\
		&\leq \frac{1}{2} \int_0^T \hat \Ex(\sigma_t) \dd t.
\end{align*}
So it suffices to determine a bound on $\hat \Ex(\sigma_t)$ for $t \leq T$. Under $\hat \Qro$, we have
\begin{align*}
	\dd \sigma_t = \Big [\rho_t\beta(t, \sigma_t) \sqrt{\sigma_t} \rind{\{t \leq \tau_n \}}  + \alpha(t, \sigma_t) \Big] \dd t + \beta(t, \sigma_t) \dd \hat B_t.
\end{align*}
Now as a consequence of \Cref{ass:martingale}, this implies that there exists some constant $M>0$ such that
\begin{align}
	\rho_t \beta(t,x) \sqrt{x} + \alpha(t,x) \leq M(1 + x), \label{importantbound1}
\end{align}
for all $x \geq 0$, uniformly in $t \leq T$. Utilising \cref{importantbound1} yields
\begin{align*}
	\dd \sigma_t \leq M (1 + \sigma_t) \dd t + \beta(t, \sigma_t) \dd \hat B_t.
\end{align*}
Hence
\begin{align*}
	\hat \Ex (\sigma_t) &\leq  \hat \Ex \left ( \sigma_0 +M \int_0^t (1+\sigma_u) \dd u + \int_0^t \beta(u, \sigma_u) \dd \hat B_u \right ) \\
				&= \sigma_0 + M \int_0^t (1 + \hat \Ex(\sigma_u)) \dd u.
\end{align*}
Define $m_t := \hat \Ex (\sigma_t)$. After redefining constants, we obtain the following integral inequality,
\begin{align*}
	m_t \leq c_0(1 + t) + c\int_0^t m_u \dd u.
\end{align*} 
Then by Gronwall's inequality, we obtain
\begin{align*}
	m_{t} \leq c_0(1 + t) e^{c t}.
\end{align*} 
\end{proof}

%----------------------------------------------------------------------------Pricing a put option---------------------------------------------------------------------------
Let 
\begin{align*}
	\text{Put} := e^{-\int_0^T r_t^d \dd t } \Ex ( K - S_T)_+
\end{align*}
be the price of a put option on $S$. Denote by $(\FF_t^B)_{0 \leq t \leq T}$ the filtration generated by the Brownian motion $B$, as well as $\NN(\cdot)$ and $\phi(\cdot)$ the standard Normal distribution and density functions respectively.
%------Prop: Mixing sol spot
\begin{proposition}
$\text{Put}$ can be expressed as
\begin{align}
	\begin{split}
		\text{Put} &= \Ex \left (  e^{-\int_0^T r_t^d \dd t } \Ex \big  [ (K - S_T)_+ | \FF_T^B \big ] \right ) \\
		&= \Ex \left ( \text{Put}_{\text{BS}} \left ( S_0 \xi_T, \int_0^T \sigma_t ( 1 - \rho_t^2) \dd t \right ) \right ), \label{mixing}
	\end{split}
\end{align}
where 
\begin{align}
\begin{split}
	\text{Put}_{\text{BS}}(x, y) &:= K e^{-\int_0^T r_t^d \dd t } \NN (-d_- ) - x e^{-\int_0^T r_t^f \dd t } \NN ( -d_+),  \label{PBS}  \\
	d_{\pm}(x,y) := d_{\pm} &:= \frac{ \ln(x/K) + \int_0^T ( r_t^d - r_t^f ) \dd t}{\sqrt{y}} \pm \frac{1}{2} \sqrt{y}.
\end{split}
\end{align}
\end{proposition}
\begin{proof}
This is a consequence of the mixing solution methodology, which is detailed in \Cref{appen:mixing}.
\end{proof}

%-----------------------------------------------------------------------------------Expansion----------------------------------------------------------------------------------
%-------Prop: Second-order put option price approx
\begin{proposition}[Second-order put option price approximation]
\label{prop:secondorderput}
The second-order put option price approximation, denoted by $\text{Put}^{(2)}$, is given by
\begin{align}
\begin{split}
\text{Put}^{(2)} &=  \text{Put}_{\text{BS}} ( \hat x, \hat y) \\ &+ \frac{1}{2} \partial_{xx} \text{Put}_{\text{BS}} ( \hat x, \hat y)  S_0^2 \Ex (\xi _T - 1)^2 + \frac{1}{2}\partial_{yy} \text{Put}_{\text{BS}} ( \hat x, \hat y)   \Ex \left (\int_0^T (1 - \rho_t^2) (\sigma_t - \Ex(\sigma_t)) \dd t \right )^2 \\ 
&+\partial_{xy} \text{Put}_{\text{BS}} ( \hat x, \hat y) S_0 \Ex \left \{  ( \xi_T - 1) \left (\int_0^T (1 - \rho_t^2) (\sigma_t - \Ex(\sigma_t)) \dd t \right ) \right \}, \label{price2}
\end{split}
\end{align}
where $ (\hat x, \hat y) := ( S_0, \int_0^T ( 1 - \rho_t^2) \Ex(\sigma_t) \dd t)$.
\end{proposition}
\begin{proof}
Notice that $\text{Put}_{\text{BS}}$ is a composition of smooth functions. Hence, $\text{Put}_{\text{BS}}$ is smooth on $(\reals^2_+;\reals_+)$. We Taylor expand $\text{Put}_{\text{BS}}$ around the mean of $ \left ( S_0 \xi_T, \int_0^T \sigma_t ( 1 - \rho_t^2) \dd t \right )$ up to second-order. By \Cref{lemma:martingalexi}, the expansion point is $ (\hat x, \hat y) = ( S_0, \int_0^T ( 1 - \rho_t^2) \Ex(\sigma_t) \dd t)$. Thus
\begin{align*}
	&\text{Put}_{\text{BS}} \left ( S_0 \xi_T, \int_0^T \sigma_t  ( 1 - \rho_t^2) \dd t \right ) \approx \  \text{Put}_{\text{BS}} ( \hat x, \hat y) \\&
	+ \partial_x \text{Put}_{\text{BS}} ( \hat x, \hat y) S_0 ( \xi_T - 1) + \partial_y \text{Put}_{\text{BS}} ( \hat x, \hat y)   \left (\int_0^T (1 - \rho_t^2) (\sigma_t - 		\Ex(\sigma_t)) \dd t \right ) \\
	&+ \frac{1}{2} \partial_{xx} \text{Put}_{\text{BS}} ( \hat x, \hat y)  S_0^2 (\xi _T - 1)^2 + \frac{1}{2}\partial_{yy} \text{Put}_{\text{BS}} ( \hat x, \hat y)    \left 		(\int_0^T (1 - \rho_t^2) (\sigma_t - \Ex(\sigma_t)) \dd t \right )^2 \\ 
	&+\partial_{xy} \text{Put}_{\text{BS}} ( \hat x, \hat y) S_0 ( \xi_T - 1) \left (\int_0^T (1 - \rho_t^2) (\sigma_t - \Ex(\sigma_t)) \dd t \right ).
\end{align*}
Taking expectation yields a second-order approximation to the put option price, that is, $\text{Put}^{(2)}$. Notice that $\text{Put}_{\text{BS}} ( \hat x, \hat y) $ is a deterministic quantity, thus the first-order terms will vanish. 
\end{proof}
The partial derivatives of $\text{Put}_{\text{BS}}$ are analogous to the Black-Scholes Greeks, which are explicit and provided in \Cref{appen:greeks}. In order to simplify $\text{Put}^{(2)}$, what remains to be done is the calculation of each of the expectations from \Cref{prop:secondorderput}, which are 
\begin{align} 
	&\Ex (\xi _T - 1)^2, \label{mom1} \\
	&\Ex \left (\int_0^T (1 - \rho_t^2) (\sigma_t - \Ex(\sigma_t)) \dd t \right )^2, \label{mom2} \\
	&\Ex \left \{  ( \xi_T - 1) \left (\int_0^T (1 - \rho_t^2) (\sigma_t - \Ex(\sigma_t)) \dd t \right ) \right \}. \label{mom3}
\end{align}

%---Remark: Call options
\begin{remark}[Call options]
When considering the pricing of call options, the approximation methodology is essentially the same as that of put options. It is clear that the mixing solution methodology from \Cref{appen:mixing} can be easily adapted to the case of call options, in which case we obtain
\begin{align*}
\text{Call} &:= e^{-\int_0^T r_t^d \dd t }  \Ex (  S_T- K)_+  =  \Ex \left (  e^{-\int_0^T r_t^d \dd t } \Ex \big  [ ( S_T-K)_+ | \FF_T^B \big ] \right ) \\
&= \Ex \left ( \text{Call}_{\text{BS}} \left ( S_0 \xi_T, \int_0^T \sigma_t ( 1 - \rho_t^2) \dd t \right ) \right ),
\end{align*}
where 
\begin{align*}
\text{Call}_{\text{BS}}(x, y) &:=  x e^{-\int_0^T r_t^f \dd t } \NN ( d_+) - K e^{-\int_0^T r_t^d \dd t } \NN (d_- ).
\end{align*}
Moreover, put-call parity states that
\begin{align*} 
	\text{Call}_{\text{BS}}(x,y) -\text{Put}_{\text{BS}}(x,y) = x e^{-\int_0^T r_t^f  \dd t} - K e^{-\int_0^T r_t^d \dd t }.
\end{align*}
Thus the second-order Call partial derivatives will be the same as the second-order Put partial derivatives. Comparing the mixing solution expressions for Call and Put, one can see that the only difference between the call and put option second-order approximation is the zero-th order term. For this reason, in this article, we need only consider the pricing of put options.
\end{remark}

%---Remark: Greeks
\begin{remark}[Greeks]
One can obtain second-order approximations of put option Greeks by simply differentiating the expression for $\text{Put}^{(2)}$. These expressions are provided in \Cref{appen:Greeks}.
\end{remark}

%--------------------------------------------------------------------------Calculation of moments-----------------------------------------------------------------------
%---------------------------------------------------------------------------------------------------------------------------------------------------------------------------------
\section{Calculation of expectations}
\noindent
\label{sec:mom}
%-------Lemma: Girsanov theorem for sequence of measures 
\begin{lemma}
\label{lem:girsseqmeas}
Let $(\Qro_n)_{n \geq 0}$ be a sequence of probability measures equivalent to $\Qro$, defined by the Radon-Nikodym derivatives
\begin{align*}
	\frac{\dd \Qro_{n+1} }{ \dd \Qro_{n}} := \xi_T^{{(n)}} :=  \exp \left \{ \int_0^T \rho_u \sqrt{\sigma_u} \dd B^{n}_u - \frac{1}{2} \int_0^T \rho_u^2 \sigma_u \dd u 	\right \},\quad \xi_T^{(0)} := \xi_T, \quad n \geq 0,
\end{align*}
where $\Qro_0 := \Qro$ and $B^0 := B$. Under $\Qro_n$, $B_t^n := B_t^{n-1} - \int_0^t \rho_u \sqrt{\sigma_u} \dd u $ is a Brownian motion. 
\end{lemma}
\begin{proof}
This is a clear consequence of Girsanov's theorem.
\end{proof}

%-------Remark: Girsanov theorem for sequence of measures 
\begin{remark}
\label{rem:girsseqmeas}Through \Cref{lem:girsseqmeas}, we have the following relationships:
\begin{align}
	\xi_T^{(n)} = \xi_T^{(n-1)} e^{-\int_0^T \rho_u^2 \sigma_u \dd u }, \quad n \geq 1, \label{calc1}
\end{align}
and
\begin{align}
\begin{split}
	\Ex_{\Qro_n}(X) &= \Ex_{\Qro_{n-1}}(X \xi_T^{(n-1)}), \\
	\Ex_{\Qro_{n-1}}(X) &= \Ex_{\Qro_n}\left (X \frac{1}{\xi_T^{(n-1)}} \right ). \label{calc2}
\end{split}
\end{align}
These relationships allow for alternative and often more convenient calculations of expectations under $\Qro$.
\end{remark}

%----------------------------------------------------------------------------Moment calcuations------------------------------------------------------------------------
Using \Cref{rem:girsseqmeas}, we can now give alternative expressions for the expectations seen in \crefrange{mom1}{mom3}.

%----------------------------------------------------------------------------------Moment 1--------------------------------------------------------------------------------
\subsection{\texorpdfstring{$\Ex(\xi_T - 1)^2$}{First expectation}}
%\pdfbookmark[2]{First expectation}{Firstexpectation}
\label{sec:mom1}
First, expanding \cref{mom1} gives 
\begin{align*}
\Ex(\xi_T - 1)^2 = \Ex(\xi_T^2) - 1.
\end{align*}
This second moment can be dealt with a number of changes of measures.
\begin{align}
	\Ex(\xi_T^2) &= \Ex_{\Qro_1}(\xi_T) = \Ex_{\Qro_1}(\xi_T^{(1)}e^{\int_0^T \rho_t^2 \sigma_t \dd t }) \nonumber \\
			   &= \Ex_{\Qro_2} ( e^{\int_0^T \rho_t^2 \sigma_t \dd t }). \label{Exexp2}
\end{align}
Under the assumption of constant parameters\footnote{That is, $\alpha(t, \sigma_t) = \alpha(\sigma_t)$, $\beta(t, \sigma_t) = \beta(\sigma_t)$ and $\rho_t = \rho$. } we may calculate \cref{Exexp2} explicitly via the Laplace transform for certain processes $\sigma$. However to our knowledge, there exists no explicit solution when parameters are time-dependent, see \citep{hurd2008explicit}. Instead, we approximate \cref{Exexp2} by expanding the exponential around the mean of $\int_0^T \rho_t^2 \sigma_t \dd t $ to second-order. 
\begin{align*}
&\Ex_{\Qro_2} (e^{\int_0^T \rho_t^2 \sigma_t \dd t }) \\ &\approx  \Ex_{\Qro_2} \left \{ e^{\int_0^T \rho_t^2  \Ex_{\Qro_2} (\sigma_t) \dd t } \left[ 1 +   \int_0^T \rho_t^2  \left (\sigma_t - \Ex_{\Qro_2} (\sigma_t) \right ) \dd t  + \frac{1}{2} \left ( \int_0^T \rho_t^2  \left (\sigma_t - \Ex_{\Qro_2} (\sigma_t)  \right ) \dd t \right )^2 \right ]\right \}\\
&=  e^{\int_0^T \rho_t^2  \Ex_{\Qro_2} (\sigma_t) \dd t} \left  \{ 1 + \frac{1}{2}  \Ex_{\Qro_2} \left [ \left ( \int_0^T \rho_t^2  \left (\sigma_t - \Ex_{\Qro_2} (\sigma_t) \right ) \dd t  \right )^2 \right ] \right  \} \\
&=  e^{\int_0^T \rho_t^2  \Ex_{\Qro_2} (\sigma_t) \dd t} \left  \{ 1 + \int_0^T \rho_t^2 \int_0^t \rho_s^2 \Cov_{\Qro_2}(\sigma_s, \sigma_t) \dd s \dd t  \right  \},
\end{align*}
where we have used the fact that $\left (\int_0^T f(t) \dd t \right )^2 = 2 \int_0^T  f(t) \left  ( \int_0^t f(s) \dd s \right ) \dd t$. 

%----------------------------------------------------------------------------------Moment 2--------------------------------------------------------------------------------
\subsection{\texorpdfstring{$\Ex \left (\int_0^T (1 - \rho_t^2) ( \sigma_t  - \Ex( \sigma_t )) \dd t \right )^2$}{Second expectation}}
\label{sec:mom2}
To calculate \cref{mom2}, we use the same approach from \Cref{sec:mom1}.
\begin{align*}
	\Ex \left (\int_0^T (1 - \rho_t^2) (\sigma_t - \Ex(\sigma_t)) \dd t \right )^2 &= 2 \int_0^T (1 - \rho_t^2) \left ( \int_0^t (1 - \rho_s^2) \Cov(\sigma_s, \sigma_t) \dd s \right ) \dd t.
\end{align*}

%----------------------------------------------------------------------------------Moment 3--------------------------------------------------------------------------------
\subsection{\texorpdfstring{$\Ex \left \{  ( \xi_T - 1) \left (\int_0^T (1 - \rho_t^2) ( \sigma_t  - \Ex( \sigma_t )) \dd t \right ) \right \}$}{Third expectation}}
\label{sec:mom3}
Calculation of the mixed expectation \cref{mom3} gives 
\begin{align*}
\Ex \left \{  ( \xi_T - 1) \left (\int_0^T (1 - \rho_t^2) (\sigma_t - \Ex(\sigma_t)) \dd t \right ) \right \}  &= 
  \int_0^T (1 - \rho_t^2)\big ( \Ex(\xi_T \sigma_t) - \Ex(\sigma_t) \big )  \dd t \\
&= \int_0^T (1 - \rho_t^2)\big ( \Ex_{\Qro_1}( \sigma_t) - \Ex(\sigma_t) \big )  \dd t.
\end{align*}

%--------------------------------------------------------------------------Pricing for specific models-------------------------------------------------------------------------
%--------------------------------------------------------------------------------------------------------------------------------------------------------------------------------------
\section{Pricing for specific models}
\noindent
\label{sec:pricingspecmodels2}
We now introduce specific dynamics for both the spot and its underlying variance process. From \Cref{sec:mom}, it is apparent that simplifying the expression for $\text{Put}^{(2)}$ will largely depend on the tractability of the variance process $\sigma$ under the original measure $\Qro$, as well as the artificial measures $\Qro_1$ and $\Qro_2$. 

%-------Remark: Notation for key processes
\begin{remark}[Notation for key processes]
Let $(\tilde \kappa_t)_{0 \leq t \leq T}, (\tilde \theta_t)_{0 \leq t \leq T}$ and $(\tilde \lambda_t)_{0 \leq t \leq T}$ be time-dependent, deterministic, strictly positive and bounded on $[0,T]$. Let $\tilde B$ be an arbitrary Brownian motion. We will utilise the following terminology:
\begin{enumerate}[label = (\arabic*), ref = \arabic*]
\item If $\tilde V$ solves 
\begin{align*}
\dd \tilde V_t = \tilde \kappa_t (\tilde \theta_t - \tilde V_t) \dd t + \tilde \lambda_t \sqrt{\tilde V_t} \dd \tilde B_t, \quad \tilde V_0 =  \tilde v_0,
\end{align*}
then we call $\tilde V$ a CIR$(\tilde v_0 ; \tilde \kappa_t, \tilde \theta_t, \tilde \lambda_t)$.
\item If $\tilde V$ solves 
\begin{align*}
	\dd \tilde V_t = \tilde \kappa_t ( \tilde \theta_t - \tilde V_t) \dd t + \tilde \lambda_t \tilde V_t \dd \tilde B_t, \quad \tilde V_0 = \tilde v_0,
\end{align*}
then we call $\tilde V$ an IGa$(\tilde v_0 ; \tilde \kappa_t, \tilde \theta_t, \tilde \lambda_t)$.
\end{enumerate}
We point the reader towards \Cref{appen:moments} for further information on these processes.
\end{remark}

%------------------------------------------------------------------------------Heston model-----------------------------------------------------------------------------------
\subsection{Heston model} 
Suppose the spot $S$ with variance $V$ follows the Heston dynamics 
\begin{align}
\begin{split}
\label{HESTON}
\dd S_t &= S_t ( (r_t^d - r_t^f ) \dd t + \sqrt{V_t} \dd W_t ), \quad S_0, \\
\dd V_t &= \kappa_t (\theta_t - V_t) \dd t + \lambda_t \sqrt{V_t}  \dd B_t , \quad V_0 = v_0,\\
\dd \langle W, B \rangle_t &= \rho_t \dd t,
\end{split}
\end{align}
where $(\kappa_t)_{0 \leq t \leq T}, (\theta_t)_{0 \leq t \leq T}$ and $(\lambda_t)_{0 \leq t \leq T}$ are time-dependent, deterministic, strictly positive and bounded on $[0, T]$.

It is clear that the variance process $V$ in \cref{HESTON} is a CIR$(v_0 ; \kappa_t , \theta_t, \lambda_t)$. 
%-------Remark: Heston xi is a martingale
\begin{remark}
Note that the Heston model satisfies \Cref{ass:martingale}, since
\begin{align*}
	\sup_{t\in [0,T]} \limsup_{x \to \infty} \frac{\rho_t \lambda_t \sqrt{x} \sqrt{x} + \kappa_t (\theta_t - x)}{x} = 	\sup_{t\in [0,T]} \limsup_{x \to \infty} \left (\rho_t \lambda_t - \kappa_t + \frac{\kappa_t \theta_t}{x} \right ) < \infty.
\end{align*}
Hence by \Cref{lemma:martingalexi}, $\xi$ is a martingale.
\end{remark}

%-------Lemma: Heston change of measure 
\begin{lemma}
Let $(\Qro_n)_{n \geq 0}$ be a sequence of probability measures equivalent to $\Qro$, defined by the Radon-Nikodym derivatives
\begin{align*}
	\frac{\dd \Qro_{n+1} }{ \dd \Qro_{n}} := \xi_T^{{(n)}} :=  \exp \left \{ \int_0^T \rho_u \sqrt{V_u} \dd B^{n}_u - \frac{1}{2} \int_0^T \rho_u^2 V_u \dd u \right \},\quad \xi_T^{(0)} := \xi_T, \quad n \geq 0,
\end{align*}
where $\Qro_0 := \Qro$ and $B^0 := B$. Under $\Qro_n$, $B_t^n := B_t^{n-1} - \int_0^t \rho_u \sqrt{V_u} \dd u $ is a Brownian motion. For $ n \geq 0$, the dynamics of $V$ under the measure $\Qro_n$ are
\begin{align*}
	\dd V_t = (\kappa_t - n\lambda_t \rho_t)\left  ( \frac{\theta_t \kappa_t}{\kappa_t - n\lambda_t \rho_t }  - V_t \right) \dd t + \lambda_t \sqrt{V_t} \dd B_t^n,
\end{align*}
which is a CIR$(v_0; \kappa_t - n\lambda_t \rho_t, \frac{\theta_t \kappa_t}{\kappa_t - n\lambda_t \rho_t } , \lambda_t)$. 
\end{lemma}
\begin{proof}
This lemma is simply obtained through \Cref{lem:girsseqmeas}, then expressing $V$ under the new measures $\Qro_n$. 
\end{proof}
Thus the variance process $V$ is a CIR process under all measures considered.

%----------------------------------------------------------------Pricing under Heston framework
\subsubsection{Pricing under the Heston framework}
\label{pricingH} The second-order approximation of the put option price in the Heston framework is given by the following theorem. 
%-------Thm: Second-order Heston price approximation
\begin{theorem}[Second-order Heston put option price]
\label{thm:hestonprice}
Let $\hat x = S_0$ and 
\begin{align*}
	\hat y = \int_0^T ( 1 - \rho_t^2) \Ex(V_t) \dd t = \int_0^T (1- \rho_t^2) \left \{ v_0 e^{- \int_0^t \kappa_z \dd z } + \int_0^t e^{- \int_u^t \kappa_z \dd z } \kappa_u \theta_u \dd u \right \} \dd t.
\end{align*}
The second-order approximation to the put option price in the Heston model, denoted by $\text{Put}_{\text{H}}^{(2)}$, is
\begin{align}
\begin{split}
\text{Put}^{(2)}_{\text{H}} &=  \text{Put}_{\text{BS}} ( \hat x, \hat y) \\ &+ \frac{1}{2} \partial_{xx} \text{Put}_{\text{BS}} ( \hat x, \hat y)  S_0^2 \Ex (\xi _T - 1)^2 + \frac{1}{2}\partial_{yy} \text{Put}_{\text{BS}} ( \hat x, \hat y)   \Ex \left (\int_0^T (1 - \rho_t^2) (V_t - \Ex(V_t)) \dd t \right )^2 \\ 
&+\partial_{xy} \text{Put}_{\text{BS}} ( \hat x, \hat y) S_0 \Ex \left \{ (\xi_T - 1) \left (\int_0^T (1 - \rho_t^2) (V_t - \Ex(V_t)) \dd t \right ) \right \}, \label{hprice2}
\end{split}
\end{align}
where
\begin{align*}
 \hspace{-.8cm} \Ex (\xi _T - 1)^2 &\approx e^{\int_0^T \rho_t^2  \Ex_{\Qro_2} (V_t) \dd t } \left  \{ 1 +  \int_0^T \rho_t^2 \int_0^t \rho_s^2 \Cov_{\Qro_2} (V_s, V_t) \dd s \dd t \right  \} -1,\\
 \hspace{-.8cm} \Ex_{\Qro_2}(V_t) &= v_0 e^{-\int_0^t \kappa_z - 2 \lambda_z \rho_z \dd z } + \int_0^t e^{-\int_u^t \kappa_z - 2\lambda_z \rho_z \dd z } \kappa_u \theta_u \dd u, \\
 \hspace{-.8cm} \Cov_{\Qro_2}(V_s, V_t) &= e^{-\int_s^t \kappa_z - 2\lambda_z \rho_z \dd z } \\ &\cdot \int_0^s \lambda_u^2 e^{-2\int_u^s \kappa_z - 2 \lambda_z \rho_z \dd z } \left [ v_0 e^{-\int_0^u \kappa_z - 2 \lambda_z \rho_z \dd z } + \int_0^u e^{-\int_p^u \kappa_z - 2 \lambda_z \rho_z \dd z } \kappa_p \theta_p \dd p \right ] \dd u,
\end{align*}
\begin{align*}
	\Ex \left (\int_0^T (1 - \rho_t^2) (V_t - \Ex(V_t)) \dd t \right )^2 
	 &= 2 \int_0^T (1 - \rho_t^2) \Bigg ( \int_0^t (1 - \rho_s^2)   \\ &\hspace{-3cm}\cdot \left [  e^{-\int_s^t \kappa_z \dd z } \int_0^s \lambda_u^2 e^{-2 \int_u^s \kappa_z \dd z} \left \{ v_0 e^{- \int_0^u \kappa_z \dd z } + \int_0^u e^{- \int_p^u \kappa_z \dd z }\kappa_p \theta_p \dd p \right \} \dd u \right ]     \dd s \Bigg ) \dd t,
\end{align*}
and
\begin{align*}
&\Ex \left \{  ( \xi_T - 1) \left (\int_0^T (1 - \rho_t^2) (V_t - \Ex(V_t)) \dd t \right ) \right \} \\
&= \int_0^T (1 - \rho_t^2) \Bigg \{ v_0 \left (e^{-\int_0^t \kappa_z -  \lambda_z \rho_z  \dd z}  - e^{-\int_0^t \kappa_z  \dd z } \right ) + \int_0^t \left (e^{-\int_u^t \kappa_z -  \lambda_z \rho_z  \dd z}  - e^{-\int_u^t \kappa_z  \dd z } \right )   \kappa_u \theta_u \dd u \Bigg \} \dd t.
\end{align*}
\begin{proof}
Use \Cref{prop:secondorderput} and adapt \Cref{sec:mom} to the Heston framework. The moments of $V$ correspond to the moments of a CIR process, which are obtained from \Cref{CIRappend}.
\end{proof}
\end{theorem}

%--------------------------------------------------------------------------GARCH diffusion model---------------------------------------------------------------------------
\subsection{GARCH diffusion model}
Suppose the spot  $S$ with variance $V$ follows the GARCH diffusion dynamics 
\begin{align}
\begin{split}
	\dd S_t &= S_t ( (r_t^d - r_t^f ) \dd t + \sqrt{V_t} \dd W_t ), \quad S_0, \\
	\dd V_t &= \kappa_t (\theta_t - V_t) \dd t + \lambda_t V_t  \dd B_t , \quad V_0 = v_0, \label{GARCH}\\
	\dd \langle W, B \rangle_t &= \rho_t \dd t, 
\end{split}
\end{align}
where $(\kappa_t)_{0 \leq t \leq T}, (\theta_t)_{0 \leq t \leq T}$ and $(\lambda_t)_{0 \leq t \leq T}$ are time-dependent, deterministic, strictly positive and bounded on $[0, T]$. It is evident that the variance process $V$ in \cref{GARCH} is an IGa$(v_0; \kappa_t, \theta_t, \lambda_t)$.

%-------Remark: GARCH xi not martingale
\begin{remark}
Unlike the Heston model, the GARCH diffusion model does not satisfy \Cref{ass:martingale}, since
\begin{align*}
	\sup_{t\in [0,T]} \limsup_{x \to \infty} \frac{\rho_t \lambda_t x \sqrt{x} + \kappa_t (\theta_t - x)}{x} = 	\sup_{t\in [0,T]} \limsup_{x \to \infty} \left (\rho_t \lambda_t\sqrt{x} - \kappa_t + \frac{\kappa_t \theta_t}{x} \right ) = \infty.
\end{align*}
Thus, there is no guarantee that $\xi$ is indeed a martingale. Instead, we will now assume that $\xi$ is a martingale and show that, even if it is indeed a martingale, it would provide no benefit in our expansion procedure.
\end{remark}

%-------Lemma: Change of measure GARCH diffusion
\begin{lemma}
\label{igalem}
Let $(\Qro_n)_{n \geq 0}$ be a sequence of probability measures equivalent to $\Qro$, defined by the Radon-Nikodym derivatives
\begin{align*}
	\frac{\dd \Qro_{n+1} }{ \dd \Qro_{n}} := \xi_T^{{(n)}} :=  \exp \left \{ \int_0^T \rho_u \sqrt{V_u} \dd B^{n}_u - \frac{1}{2} \int_0^T \rho_u^2 		V_u \dd u \right \},\quad \xi_T^{(0)} := \xi_T, \quad n \geq 0,
\end{align*}
where $\Qro_0 := \Qro$ and $B^0 := B$. Under $\Qro_n$, $B_t^n := B_t^{n-1} - \int_0^t \rho_u \sqrt {V_u} \dd u $ is a Brownian motion. For $ n \geq 0$, the dynamics of $V$ under the measure $\Qro_n$ are
\begin{align*}
	\dd V_t &=  \kappa_t \left (\theta_t  - V_t + \frac{n \lambda_t \rho_t }{\kappa_t} V_t^{3/2} \right ) \dd t  + \lambda_t V_t \dd B^n_t.
\end{align*} 
\end{lemma}
\begin{proof}
This lemma is simply obtained through \Cref{lem:girsseqmeas}, then expressing $V$ under the new measures $\Qro_n$. 
\end{proof}

%-------Remark: IGa process has no explicit solution or moments under Q_n 
\begin{remark}
\label{noexplicitsoliga}
Let $\Qro_n$ be defined as in \Cref{igalem}. Under the measures $\Qro_n$, $n \geq 1$, $V$ has no known expression for its solution, nor known expression for its moments.
\end{remark}  
\begin{proof}[Validity of \Cref{noexplicitsoliga}]
The SDE in \Cref{igalem} is a linear diffusion type SDE. From \Cref{appen:SDElinear}, it is known that if an explicit solution exists, it is given by 
\begin{align*}
	V_t = Y_t/F_t, 
\end{align*}
where $F$ is a GBM$(1; \lambda_t^2, -\lambda_t)$ and $Y$ is the solution to the integral equation (written in differential form)
\begin{align*}
	\dd Y_t = \left (\kappa_t \theta_t F_t - \kappa_t Y_t + \frac{n \lambda_t \rho_t}{\kappa_t} Y_t^{3/2} F_t^{-1/2} \right ) \dd t.
\end{align*}
Define $A_t := \kappa_t \theta_t F_t$ and $C_t := \frac{n \lambda_t \rho_t}{\kappa_t} F_t^{-1/2} $. Note that $A_t$ and $C_t$ are both non-differentiable in $t$. Thus 
\begin{align*}
	\dd Y_t = \left (A_t - \kappa_t Y_t + C_t Y_t^{3/2} \right ) \dd t.
\end{align*}
As far as we know, there is no explicit solution to these types of integral equations in the literature, even when $A$ and $C$ are constants. As for explicit moments, it is unclear how to approach this problem. There seems to be no approach to this problem in the literature, especially in the case of time-dependent parameters, see for example \citep{kloeden2013numerical} Chapter 4.4 for a comprehensive list of explicitly solvable SDEs. Furthermore, as an explicit solution does not exist, we cannot use the method of approximating moments via the SDE's solution.
\end{proof}

%-----------------------------------------------Pricing under GARCH diffusion Framework: rho = 0
\subsubsection{Pricing under the GARCH diffusion framework: $\rho = 0$}
In the GARCH diffusion model, \Cref{ass:martingale} is violated, and thus there is no guarantee $\xi$ is a martingale. Additionally, assuming $\xi$ is a martingale is not helpful, as the change of measure technique gives an intractable dynamic for $V$; we cannot appeal to it for calculating expectations. However, in the case of $\rho = 0$ a.e., this implies $\xi_T = 1$ $\Qro$ a.s., and one will notice that the terms in the expansion requiring a change of measure will disappear. Of course, the cost is the additional restrictive assumption that spot and volatility movements are uncorrelated. We hope to mitigate this issue in future work by combining this approach with small correlation expansion methods, see \citep{Antonelli09,Antonelli10}.

%-------Thm: GARCH second-order price approx
\begin{theorem}[Second-order GARCH put option price]
\label{thm:GARCHprice}
Assume $\rho = 0$ a.e. Let $\hat x = S_0$ and $\hat y = \int_0^T \Ex(V_t) \dd t  $. Then the second-order put option price in the GARCH diffusion model, denoted by $\text{Put}^{(2)}_{\text{GARCH}}$, is
\begin{align}
\begin{split}
\text{Put}^{(2)}_{\text{GARCH}} &=  \text{Put}_{\text{BS}} ( \hat x, \hat y) + \partial_{yy} \text{Put}_{\text{BS}} ( \hat x, \hat y)  \int_0^T\left ( \int_0^t \Cov(V_s, V_t) \dd s \right ) \dd t. \label{Garprice2}
\end{split}
\end{align}
Both $\Cov(V_s, V_t)$ and $\Ex(V_t)$ are given in \Cref{appen:IGa}.
\end{theorem}
\begin{proof}
Use \Cref{prop:secondorderput} under the assumption of $\rho = 0$ a.e. Then notice that
\begin{align*}
 \Ex \left (\int_0^T  (V_t - \Ex(V_t)) \dd t \right )^2 &= 2 \int_0^T\left ( \int_0^t \Cov(V_s, V_t) \dd s \right ) \dd t.
\end{align*} 
\end{proof}

\section{Error analysis}
\noindent
\label{sec:erroranalysis2}We present an explicit bound on the error term in our expansion in terms of higher order moments of the corresponding variance process. Specifically, this means bounding the remainder term in the second-order expansion of the function Put$_{\text{BS}}$, and for the case when $\rho \neq 0$, the error term associated with the expansion of $e^{\int_0^T \rho_u^2 \sigma_u \dd u }$. 

We will need explicit expressions for the error terms. These are given by Taylor's theorem, which we will present here to fix notation. We only consider the results up to second-order. 

%-----Taylor's theorem
\begin{theorem}[Taylor's theorem for $f: \reals \to \reals$]
\label{tay}
Let $A \subseteq \reals, B \subseteq \reals$ and $f : A \to B$ be a $C^{3}$ function in a closed interval about the point $a \in A$. Then
\begin{align*}
	f(x)&= f(a)+ f'(a)(x-a) + \frac{1}{2} f''(a) (x-a)^2 + R(x),
\end{align*}
where the remainder $R$ is given by
\begin{align*}
	R(x) =\frac{1}{2} \int_a^x (x - u)^2 f'''(u) \dd u =  \frac{1}{2} (x -a)^3 \int_0^1 (1 - u)^2 f'''(a + u(x - a)) \dd u.
\end{align*}
\end{theorem}
We will prefer the integration bounds to be from $0$ to $1$ rather than $a$ to $x$, as $a$ and $x$ will correspond to random variables.

%-----Taylor's theorem in 2D
\begin{theorem}[Taylor's theorem for $g: \reals^2 \to \reals$]
\label{tay2}
Let $A \subseteq \reals^2, B \subseteq \reals$ and $g : A \to B$ be a $C^3$ function in a closed ball about the point $(a, b) \in A$. Then 
\begin{align*}
g(x, y) &= g(a, b) + g_x (a, b)(x-a) + g_y(a, b) (y-b) \\&+ \frac{1}{2} g_{xx} (a,b) (x-a)^2 + \frac{1}{2} g_{yy}(a, b) (y-b)^2  + g_{xy} (a, b)(x-a)(y-b) + R(x,y),
\end{align*}
where the remainder $R$ is given by
\begin{align*}
	R(x,y) &= \sum_{|\alpha| = 3} \frac{|\alpha|}{\alpha_1! \alpha_2 ! }E_{\alpha}(x, y) ( x- a)^{\alpha_1} (y - b)^{\alpha_2}, \\
	E_\alpha (x, y) &= \int_0^1 ( 1 - u)^{2} \frac{\partial^{3} }{\partial x^{\alpha_1} \partial y^{\alpha_2}} g(a + u(x - a), b  + u(y - b)) \dd u, 
\end{align*}
with $\alpha := (\alpha_1, \alpha_2)$ and $|\alpha| := \alpha_1 + \alpha_2$.
\end{theorem}

%--------------------------------------------------------------------------------Error due to Taylor Expansion--------------------------------------------------------------
\subsection{Expression for error term}
The representation for the total error generated by the expansion can be summarised by the following theorem.

%--------Thm: Total expansion error
\begin{theorem}[Total expansion error]
\label{totalerrorpart2}
As a functional of the underlying variance process $\sigma$, let $\EE_{\text{BS}}(\sigma)$ and $\mathcal{\tilde E}(\sigma)$ correspond to the error induced by the expansion of $\text{Put}_{\text{BS}}$ and $e^{\int_0^T \rho_u^2 \sigma_u \dd u}$ respectively. The error generated by Taylor expansions for a general variance process $\sigma$ is given by
\begin{align*}
\mathcal{E}(\sigma) = \mathcal{E}_{\text{BS}}(\sigma) + \mathcal{ \tilde E  }(\sigma),
\end{align*}
where 
\begin{align*}
&\mathcal{E}_{\text{BS}}(\sigma) =  \sum_{|\alpha| = 3} \frac{|\alpha|}{\alpha_1! \alpha_2 ! }E_{\alpha}\left ( S_0 \xi_T, \int_0^T \sigma_t  ( 1 - \rho_t^2) \dd t \right ) S_0^{\alpha_1} (\xi_T - 1)^{\alpha_1} \left (\int_0^T (1 - \rho_u^2) (\sigma_u - \Ex(\sigma_u)) \dd u\right )^{\alpha_2}, \\
&E_{\alpha}\left ( S_0 \xi_T, \int_0^T \sigma_t  ( 1 - \rho_t^2) \dd t \right ) = \int_0^1 ( 1 - u)^{2} \frac{\partial^{3} }{\partial x^{\alpha_1} \partial y^{\alpha_2}} \text{Put}_{\text{BS}}\left (F(u), G(u) \right ) \dd u, \\
&F(u) := S_0 + uS_0 ( \xi_T - 1), \\
&G(u) := \int_0^T ( 1 - \rho_t^2) \Ex(\sigma_t) \dd t  + u \left ( \int_0^T (1 - \rho_t^2) (\sigma_t - \Ex(\sigma_t)) \dd t \right ),
\end{align*} 
and 
 \begin{align*}
 \mathcal{\tilde E} (\sigma) &=\frac{1}{4} \partial_{xx} \text{Put}_{\text{BS}}(\hat x , \hat y) S_0^2 \xi_T^2 e^{-\int_0^T \rho_u^2 \sigma_u \dd u }\left ( \int_0^T \rho_u^2 (\sigma_u - \Ex_{\Qro_2}(\sigma_u) )\dd u \right )^3 \\ &\cdot \int_0^1 (1 - u)^2 e^{\int_0^T \rho_m^2 \Ex_{\Qro_2}(\sigma_m) \dd m } e^{u \int_0^T \rho_m^2 (\sigma_m  - \Ex_{\Qro_2}(\sigma_m)) \dd m } \dd u. 
\end{align*}
\end{theorem}
%--------Proof: Total expansion error
\begin{proof}
First, we deal with the error term associated with the function Put$_\text{BS}$, that is, $\EE_{\text{BS}}(\sigma)$. Recall the expansion of Put$_\text{BS}$ around the point $(\hat x, \hat y) := ( S_0 , \int_0^T \Ex(\sigma_t)  ( 1 - \rho_t^2) \dd t )$ evaluated at $( S_0 \xi_T, \int_0^T \sigma_t  ( 1 - \rho_t^2) \dd t )$ for a general variance process $\sigma$:
\begin{align*}
&\text{Put}_{\text{BS}} \left ( S_0 \xi_T, \int_0^T \sigma_t  ( 1 - \rho_t^2) \dd t \right ) =\  \text{Put}_{\text{BS}} ( \hat x, \hat y) \\&
+ \partial_x \text{Put}_{\text{BS}} ( \hat x, \hat y) S_0 ( \xi_T - 1) + \partial_y \text{Put}_{\text{BS}} ( \hat x, \hat y)   \left (\int_0^T (1 - \rho_t^2) (\sigma_t - \Ex(\sigma_t)) \dd t \right ) \\
&+ \frac{1}{2} \partial_{xx} \text{Put}_{\text{BS}} ( \hat x, \hat y)  S_0^2 (\xi _T - 1)^2 + \frac{1}{2}\partial_{yy} \text{Put}_{\text{BS}} ( \hat x, \hat y)    \left (\int_0^T (1 - \rho_t^2) (\sigma_t - \Ex(\sigma_t)) \dd t \right )^2 \\ 
&+\partial_{xy} \text{Put}_{\text{BS}} ( \hat x, \hat y) S_0 ( \xi_T - 1) \left (\int_0^T (1 - \rho_t^2) (\sigma_t - \Ex(\sigma_t)) \dd t \right ) + \mathcal{E}_{\text{BS}}(\sigma).
\end{align*}
Using \Cref{tay2} for the function Put$_{\text{BS}}$, this gives the error term as
\begin{align*}
&\mathcal{E}_{\text{BS}}(\sigma) =  \sum_{|\alpha| = 3} \frac{|\alpha|}{\alpha_1! \alpha_2 ! }E_{\alpha}\left ( S_0 \xi_T, \int_0^T \sigma_t  ( 1 - \rho_t^2) \dd t \right ) S_0^{\alpha_1} (\xi_T - 1)^{\alpha_1} \left (\int_0^T (1 - \rho_u^2) (\sigma_u - \Ex(\sigma_u)) \dd u\right )^{\alpha_2}, \\
&E_{\alpha}\left ( S_0 \xi_T, \int_0^T \sigma_t  ( 1 - \rho_t^2) \dd t \right ) = \int_0^1 ( 1 - u)^{2} \frac{\partial^{3} }{\partial x^{\alpha_1} \partial y^{\alpha_2}} \text{Put}_{\text{BS}}\left (F(u), G(u) \right ) \dd u, \\
&F(u) := S_0 + uS_0 ( \xi_T - 1), \\
&G(u) := \int_0^T ( 1 - \rho_t^2) \Ex(\sigma_t) \dd t  + u \left ( \int_0^T (1 - \rho_t^2) (\sigma_t - \Ex(\sigma_t)) \dd t \right ).
\end{align*}
%------Ex(\xi_T^2) error
We now investigate the error term associated with the calculation of $\Ex \xi_T^2$, that is, $\mathcal{ \tilde E}(\sigma)$. Let us look at this term without the expectation. 
\begin{align*}
\xi_T^2 =\left (\xi_T^2 e^{-\int_0^T \rho_u^2 \sigma_u \dd u } \right ) e^{\int_0^T \rho_u^2 \sigma_u \dd u }. 
\end{align*}
We expand $e^{\int_0^T \rho_u^2 \sigma_u \dd u }$ around the expectation of the exponential's argument under $\Qro_2$. Note that $\xi_T^2 e^{-\int_0^T \rho_u^2 \sigma_u \dd u }$ is the Radon-Nikodym derivative which changes measure from $\Qro $ to $\Qro_2$. Expanding to second-order gives 
\begin{align*}
 e^{\int_0^T \rho_u^2 \sigma_u \dd u } &= e^{\int_0^T \rho_u^2 \Ex_{\Qro_2}(\sigma_u) \dd u } \left ( 1 + \int_0^T \rho_u^2 (\sigma_u - \Ex_{\Qro_2}(\sigma_u)) \dd u  + \frac{1}{2} \left (\int_0^T \rho_u^2 (\sigma_u - \Ex_{\Qro_2}(\sigma_u) )\dd u \right )^2 \right) \\&+ \frac{1}{2} \left ( \int_0^T \rho_u^2 (\sigma_u - \Ex_{\Qro_2}(\sigma_u) \dd u \right )^3  \int_0^1 ( 1- u)^2 e^{\int_0^T \rho_m^2 \Ex_{\Qro_2}(\sigma_m) \dd m } e^{u \int_0^T \rho_m^2 (\sigma_m  - \Ex_{\Qro_2}(\sigma_m) )\dd m } \dd u. 
 \end{align*}
 Finally, the coefficient in front of $\xi_T^2$ in the pricing formula is $\frac{1}{2} \partial_{xx} \text{Put}_{\text{BS}} (\hat x , \hat y ) S_0^2$. Thus, the error term $\mathcal{\tilde E}(\sigma)$ can be written as
 \begin{align*}
 \mathcal{\tilde E} (\sigma) &= \frac{1}{4} \partial_{xx} \text{Put}_{\text{BS}}(\hat x , \hat y) S_0^2 \xi_T^2 e^{-\int_0^T \rho_u^2 \sigma_u \dd u }\left ( \int_0^T \rho_u^2 (\sigma_u - \Ex_{\Qro_2}(\sigma_u) )\dd u \right )^3 \\ &\cdot \int_0^1 (1 - u)^2 e^{\int_0^T \rho_m^2 \Ex_{\Qro_2}(\sigma_m) \dd m } e^{u \int_0^T \rho_m^2 (\sigma_m  - \Ex_{\Qro_2}(\sigma_m)) \dd m } \dd u.
\end{align*}
\end{proof}

%---------Corollary: Total expansion error \rho = 0
\begin{corollary}[Total expansion error: $\rho = 0$]
\label{totalerrorrho0part2}
The error due to Taylor expansions for a general variance process $\sigma$ with $\rho = 0$ a.e., denoted by $\mathcal{E}_0(\sigma)$, is given by 
\begin{align*}
\mathcal{E}_0(\sigma) &= \frac{1}{2}  \left ( \int_0^T (\sigma_t - \Ex(\sigma_t)) \dd t \right )^3  \int_0^1 (1 - u)^2 \partial_{yyy} \text{Put}_{\text{BS}} \left (S_0,   \tilde G(u) \right ) \dd u, \\
\tilde G(u) &:= \int_0^T \Ex(\sigma_t) \dd t  + u \left ( \int_0^T (\sigma_t - \Ex(\sigma_t))  \dd t\right ),
\end{align*}
which is just $\mathcal{E}_{\text{BS}}(\sigma)$ when $\rho = 0$ a.e.
\end{corollary}

%------------------------------------------------------------------------------Error bounding---------------------------------------------------------------------------------
\subsection{Bounding error term}
The hope now is to be able to bound $\Ex (\mathcal{E}(\sigma))$ in terms of the higher order moments of the variance process $\sigma$. To do this, we will show that the partial derivatives $\frac{\partial^3 \text{Put}_{\text{BS}}}{\partial x^{\alpha_1} \partial y^{\alpha_2}}(F(u), G(u))$ for $u \in (0,1)$ where $\alpha_1 + \alpha_2 = 3$, are functions of $T$ and $K$ which are bounded on $\reals^2_+$. 

%---------Lemma: Blow up of derivatives
\begin{lemma}
\label{parblow}
Consider the third-order partial derivatives of $\text{Put}_{\text{BS}}$, $\frac{\partial^3 \text{Put}_{\text{BS}}}{\partial x^{\alpha_1} \partial y^{\alpha_2}}$, where $\alpha_1 + \alpha_2 = 3$. Let $\hat f(x) := \ln(x/K) + \int_0^T \left (r_t^d - r_t^f \right ) \dd t$. Then
\begin{align*}
	 \lim_{y \downarrow 0} \left | \frac{\partial^3 \text{Put}_{\text{BS}}}{\partial x^{\alpha_1} \partial y^{\alpha_2}} \right |_{\hat f(x) = 0}  = \infty.
\end{align*}
Furthermore, this is the only case where the partial derivatives explode.
\end{lemma}
\begin{proof}
In the following, we will repeatedly let $F$ be an arbitrary polynomial of some degree, as well as $A$ to be an arbitrary constant. That is, they may be different on each use. From \Cref{appen:greeks}, it can be seen that as a function of $x$ and $y$, the third-order partial derivatives are of the form 
\begin{align}
	A \frac{\phi(d_+)}{x^n y^{m/2}}F(d_+, d_-, \sqrt{y}),  \quad n \in \mathbb{Z}, m \in \mathbb{N}, \label{pderivform}
\end{align}
where we recall 
\begin{align*}
	d_\pm = d_\pm(x,y) &= \frac{\ln(x/K) + \int_0^T \left ( r_t^d - r_t^f \right ) \dd t }{\sqrt{y}} \pm \frac{1}{2} \sqrt{y}, \\
	\phi(x) &= \frac{1}{\sqrt{2 \pi}} e^{-x^2/2}.
\end{align*}
Written in this form \cref{pderivform}, it is evident that the partial derivatives could only blow up if either $x$ or $y$ tend to $0$ or infinity. We need only look at these limits independently of the other variable. In the following, we shall say $f = \littleo{g}$ if and only if $\lim_{x \to \infty} \frac{f(x)}{g(x)} = 0$ and $f = \littleoo{g}$ if and only if $\lim_{x \downarrow 0} \frac{f(x)}{g(x)} = 0$. 
\begin{enumerate}[label = (\arabic*), ref = \arabic*]
%-----Fixed x
\item  For fixed $x$: From \cref{pderivform} the partial derivatives are of the form $A \frac{\phi(d_+)F(d_+, d_-, \sqrt{y})}{y^{m/2} }$. 
It can be shown that 
\begin{align*}
	 \phi(d_+) = A e^{-D_2 \frac{1}{y} - D_1 y},
\end{align*}
where $D_2 = \frac{1}{2} (\hat f(x))^2$ and $D_1 = 1/8$. Hence both $D_2$ and $D_1$ are non-negative. However, there will be two cases to consider, when $D_2 > 0$ or $D_2 = 0$.
	\begin{enumerate}[label = (\alph*), ref = \alph*]
	\item Suppose $D_2 > 0$, then $\hat f (x) \neq 0$. As $F$ is a polynomial in $d_+, d_-$, and $\sqrt{y}$, we can say that $F(d_+, d_-, \sqrt{y}) = \littleoo{1/y^{M_0/2}}$ and $F(d_+, d_-, \sqrt{y}) = \littleo{y^{M/2}}$ for some $M, M_0 \in \mathbb{N}$. Thus 
\begin{align*}
	\left | \frac{\phi(d_+)F(d_+, d_-, \sqrt{y})}{y^{m/2} } \right | =  |A| \frac{e^{-D_2 \frac{1}{y} - D_1 y} \littleoo{1/y^{M_{0}/2}}}{y^{m/2}}
\end{align*}
and also 
\begin{align*}
	\left | \frac{\phi(d_+)F(d_+, d_-, \sqrt{y})}{y^{m/2} } \right | =  |A| \frac{e^{-D_2 \frac{1}{y} - D_1 y} \littleo{y^{M/2}}}{y^{m/2}}.
\end{align*}
Then as $y \downarrow 0$ or $y \to \infty$, the partial derivatives tend to $0$.
	\item Suppose $D_2 = 0$, then $\hat f(x) = 0$. 	Thus $d_+ = \frac{1}{2} \sqrt{y}$ and $\phi(d_+) = A e^{-D_1 y}$. Evidently, $F(d_+, d_-, \sqrt{y}) = \sum_{i =0}^N C_i y^{i/2}$ for some $N \in \mathbb{N}$ and constants $C_0, \dots, C_N$. Thus
	\begin{align*}
		\left | \frac{\phi(d_+)F(d_+, d_-, \sqrt{y})}{y^{m/2} } \right | = |A| \frac{e^{-D_1 y} | \sum_{i =0}^NC_i y^{i/2} |}{y^{m/2}}. 
	\end{align*}
	This quantity tends to 0 as $y \to \infty$, as the exponential decay makes the polynomial growth/decay irrelevant. However, when $y \downarrow 0$, then this limit depends on the polynomial $F$. If $N > m$ and if one of the $C_0, C_1, \dots, C_m$ are non-zero then this quantity tends to $\infty$ as $y \downarrow 0$. If $N < m$, then this quantity tends to $\infty$ as $y \downarrow 0$. For each of the partial derivatives, it can be shown that either of these cases are satisfied. Thus the partial derivatives tend to $\infty$ when $y \downarrow 0$. 
	\end{enumerate}
To conclude, for fixed $x$, if $\hat f(x) = 0$, then the partial derivatives tend to $0$ if $y \to \infty$, and to $\infty$ if $y \downarrow 0$. When $\hat f(x) \neq 0$, the partial derivatives tend to $0$ if $y \downarrow 0$ or $y \to \infty$. 
%----Fixed y
\item For fixed $y$:  From \cref{pderivform}, the partial derivatives are of the form $\frac{\phi(d_+) F(d_+, d_-)}{ x^n }$. It can be shown that 
\begin{align*}
	\phi(d_+) = A x^{-E_2 \ln(x) - E_1},
\end{align*}
where $E_2 > 0$ and $E_1 \in \reals$. Furthermore, as $F$ is a polynomial in $d_+$ and $d_-$, then
\begin{align*}
	|F(d_+, d_-) | \leq \sum_{i = 0}^N |C_i| \left | \ln^i(x) \right |,
\end{align*}
for some $N \in \mathbb{N}$ and constants $C_0, \dots C_N$.
It is clear $F = \littleo{x}$ and $F = \littleoo{\ln^{N+1}(x)}$.
Thus 
\begin{align*}
	\left | \frac{\phi(d_+) F(d_+, d_-)}{ x^n } \right | = |A| x^{-E_2 \ln(x) - E_1 - n}\littleo{x}.
\end{align*}
Then as $x \to \infty$, this quantity tends to $0$. In addition
\begin{align*}
	\left | \frac{\phi(d_+) F(d_+, d_-)}{ x^n } \right | = |A| x^{-E_2 \ln(x) - E_1 - n}\littleoo{\ln^{N+1}(x)}.
\end{align*}
Then as $x \downarrow 0$ this quantity tends to 0. So for fixed $y$, the partial derivatives tend to 0 as $x \downarrow 0$ or $x \to \infty$.
\end{enumerate}  
\end{proof}

We will however, be concerned with the behaviour of $ u \mapsto \frac{\partial^{3} }{\partial x^{\alpha_1} \partial y^{\alpha_2}} \text{Put}_{\text{BS}} (F(u) , G(u))$, meaning we will have to consider both arguments simultaneously, as they are both linear functions of $u$. 

%---------Lemma: Blow up of derivatives in linear h_1, h_2
\begin{lemma}
\label{parblowu}
Consider the third-order partial derivatives of $\text{Put}_{\text{BS}}$, $\frac{\partial^3 \text{Put}_{\text{BS}}}{\partial x^{\alpha_1} \partial y^{\alpha_2}}$, where $\alpha_1 + \alpha_2 = 3$ as well as the linear functions $h_1, h_2 :[0,1] \to \reals_+$ such that $h_1(u) = u(d_1 - c_1) + c_1$ and $h_2(u) = u(d_2 - c_2) + c_2$. Assume there exists no point $a \in (0,1)$ such that 
\begin{align*}
	\lim_{u \to a} \frac{ \ln(h_1(u)/K) + \int_0^T (r_t^d - r_t^f) \dd t  }{\sqrt {h_2(u)}} = 0  \quad \text{and} \quad \lim_{u \to a} h_2(u) = 0.
\end{align*} 
Then there exists functions $M_{\alpha}$ bounded on $\reals_+^2$ such that
\begin{align*}
	\sup_{u \in (0,1)} \left | \frac{\partial^3 \text{Put}_{\text{BS}}}{\partial x^{\alpha_1} \partial y^{\alpha_2}} (h_1(u), h_2(u)) \right | =M_{\alpha}(T, K).
\end{align*}
Furthermore, the behaviour of $M_{\alpha}$ for fixed $K$ and $T$ is characterised by the functions $\zeta$ and $\eta$ respectively, where
\begin{align*}
	\zeta(T) =  \hat A  e^{-\int_0^T r_t^f \dd t } e^{ -E_2 \tilde r^2(T) } e^{- E_1 \tilde r(T)} \sum_{i=0}^n c_i \tilde r^i (T),
\end{align*}
with $\tilde r (T) :=  \int_0^T (r_t^d - r_t^f) \dd t$ and $E_2 > 0$, $E_1 \in \reals, \hat A \in \reals$, $n \in \mathbb{N}$ and $c_0, \dots, c_n$ are constants, and
\begin{align*}
	\eta(K) =\tilde A K^{-D_2 \ln(K) + D_1} \sum_{i = 0}^N C_i (-1)^i \ln^i(K), 
\end{align*}
with $D_2 > 0, D_1 \in \reals, \tilde A \in \reals, N \in \mathbb{N}$ and $C_0, \dots, C_N$ are constants.
\begin{proof}
Under the assumptions of $h_1$ and $h_2$, this supremum will be bounded by a direct application of \Cref{parblow}. Next, we need to show that $M_{\alpha}$ are bounded on $\reals_+^2$ and behave as $\zeta$ and $\eta$ for fixed $K$ and $T$ respectively. 

In the following, $A$ denotes an arbitrary constant and $F$ an arbitrary polynomial of some degree. They may be different on each use. 
\begin{enumerate}[label = (\arabic*), ref = \arabic*]
%---T Bound
\item Behaviour in $T$: Fix all variables as constant except $T$. Then we can write the partial derivatives as
\begin{align*}
	A e^{-\int_0^T r_t^f \dd t } \phi(d_+) F(d_+, d_-).
\end{align*}
Expanding and collecting terms with $T$, we can write the partial derivatives in the form
\begin{align*}
	 A  e^{-\int_0^T r_t^f \dd t } e^{ -E_2 \tilde r^2(T) } e^{- E_1 \tilde r(T)} \sum_{i=0}^n c_i \tilde r^i (T) = \zeta(T), 
\end{align*}
where $E_2 > 0$, $E_1 \in \reals$, $n \in \mathbb{N}$ and $c_0, \dots, c_n$ are constants. As $\zeta$ is a composition of polynomials and exponentials of $\tilde r(T)$, then it is bounded for any closed interval not containing $0$. Now since $\sup_{t \in [0, T]} (|r_t^d - r_t^f|) =:R < 1$ then $|\tilde r(T)| < R T$. Thus $\tilde r^i(T) = \littleo{T^i}$ and $\tilde r^i(T) = \littleoo{T^i}$. Hence $\zeta$ tends to $0$ as $T \downarrow 0$ or $T \to \infty$. Thus $\zeta$ is bounded on $\reals_+$. 
%---K Bound
\item Behaviour in $K$: Now fix all variables as constant except $K$. Then the partial derivatives can be written as 
\begin{align*}
	A \phi(d_+) F(d_+, d_-).
\end{align*}
Expanding and collecting terms with $K$, the partial derivatives can be written in the form
\begin{align*}
	 A   K^{-D_2 \ln(K) + D_1} F(\ln(1/K)) = \eta(K), 
\end{align*}
where $D_2 > 0$ and $D_1 \in \reals$. %, where $A = \frac{1}{\sqrt{2 \pi}}$, $\tilde A = $, $E_1 = \frac{1}{2y}$, $E_2 = (1 + 2 \frac{ \tilde r(T)}{T} )$.
Then writing out the polynomial explicitly
\begin{align*}
	 \eta(K) = A K^{-D_2 \ln(K) + D_1} \sum_{i = 0}^N C_i (-1)^i \ln^i(K), 
\end{align*} 
where $N \in \mathbb{N}$ and $C_0, \dots, C_N$ are constants. Thus
\begin{align*}
	 |\eta(K)| \leq |A|K^{-D_2 \ln(K) + D_1}  \sum_{i = 0}^N |\ln^i(K)|.
\end{align*} 
$\eta$ is bounded for any closed interval not containing $0$ since it is a composition of exponentials and logarithms. Then as $\ln^i(K) = \littleo{K}$ and $\ln^i(K) = \littleoo{\ln^{N+1} (K)}$, $\eta$ tends to $0$ as $K \downarrow 0$ or $K \to \infty$. Thus $\eta$ is bounded on $\reals_+$.
\end{enumerate}
\end{proof}
\end{lemma}

%---------Prop: Bound derivatives as functions of F(u), G(u)
\begin{proposition}
\label{parboundprop1}
There exists functions $M_{\alpha}$ as in \Cref{parblowu} such that 
\begin{align*}
	\sup_{u \in (0,1)} \left | \frac{\partial^{3} }{\partial x^{\alpha_1} \partial y^{\alpha_2}} \text{Put}_{\text{BS}}\left (F(u) , G(u) \right )  \right | \leq M_{\alpha}(T, K) \quad \Qro \text{   a.s.}
\end{align*}
\end{proposition}
\begin{proof}
Since $F$ and $G$ are linear functions, then from \Cref{parblowu}, this claim is immediately true if we can show that $G$ is strictly greater than $0$ for any $u \in (0,1)$ $\Qro$ a.s. Recall
\begin{align*}
	G(u) =  (1-u) \left ( \int_0^T (1 - \rho_t^2)  \Ex(\sigma_t) \dd t \right ) + u\int_0^T ( 1 - \rho_t^2) \sigma_t \dd t.
\end{align*}
$G$ corresponds to the linear interpolation of $\int_0^T ( 1- \rho_t^2) \Ex(\sigma_t) \dd t$ and $\int_0^T (1 - \rho_t^2) \sigma_t \dd t$. It is clear $\sup_{t \in [0, T]} (1 - \rho_t^2) > 0$. As $\sigma$ corresponds to the variance process, by \Cref{ass:solutionvariance} it is a non-negative process. Hence, the set $ \{t \in [0,T] : \sigma_t > 0 \}$ has non-zero Lebesgue measure. Thus, these integrals are strictly positive and hence $G$ is strictly larger than $0$ for any $u \in (0,1)$ $\Qro$ a.s. 
\end{proof}
We obtain the following corollary. 

%---------Corollary: Bound derivatives as function of S_0, \tilde G(u)
\begin{corollary}
\label{parboundcor1}
There exists a function $M$ as in \Cref{parblowu} such that
\begin{align*}
	 	\sup_{u \in (0,1)} \left | \partial_{yyy}\text{Put}_{\text{BS}}\left (S_0, \tilde G(u) \right )  \right | \leq M (T, K) \quad \Qro \text{   a.s.}.
\end{align*}
\begin{proof}
Recall
\begin{align*}
	\tilde G(u) = (1-u) \left ( \int_0^T  \Ex(\sigma_t) \dd t \right ) + u\int_0^T \sigma_t \dd t.
\end{align*}
Then by the same argument in the proof of \Cref{parboundprop1}, $\tilde G$ is strictly greater than $0$ $\Qro$ a.s. Hence by \Cref{parblowu}, the claim is true.
\end{proof}
\end{corollary}

%-------Thm: Error bounds for general sigma
\begin{theorem}[Error bounds for general $\sigma$]
\label{thm:errbound}
The error term in the pricing formula is bounded as
\begin{align*}
\left |  \Ex \left (\mathcal{E}(\sigma) \right ) \right | &\leq  \sum_{|\alpha| = 3} C_{\alpha} M_{\alpha} (T, K) T^{\alpha_2-\frac{1}{2}} S_0^{\alpha_1}  \left (\Ex(\xi_T - 1)^{2 \alpha_1} \right )^{1/2} \left \{  \int_0^T (1 - \rho_u^2)^{ 2 \alpha_2} \Ex |\sigma_u - \Ex(\sigma_u)|^{2 \alpha_2} \dd u  \right \}^{1/2} \\
&+ C\tilde M(T, K) S_0^2 T^{5/2} e^{\int_0^T \rho_m^2  \Ex_{\Qro_2}(\sigma_m) \dd m }   \left ( \Ex_{\Qro_2} e^{\int_0^T 2 \rho_m^2 |\sigma_m - \Ex_{\Qro_2}(\sigma_m) |\dd m } \right )^{1/2} \\ &\cdot \left (  \int_0^T \rho_u^{12} \Ex_{\Qro_2} |\sigma_u - \Ex_{\Qro_2}(\sigma_u) |^6\dd u \right )^{1/2},
\end{align*}
where $\tilde M (T, K) = \partial_{xx} \text{Put}_{\text{BS}}(\hat x, \hat y)$ is bounded on $\reals_+^2$ and $C =1/12 ,C_{\alpha} = \frac{1}{3} \frac{|\alpha|}{\alpha_1! \alpha_2!}$ are constants, the latter depending on $\alpha$.
\end{theorem}
\begin{proof}
First, by \Cref{parboundprop1}, we have that $E_{\alpha}(S_0 \xi_T , \int_0^T \sigma_t(1 - \rho_t^2)) \leq \frac{1}{3} M_{\alpha}(T, K)$. By \Cref{totalerrorpart2}, the error is decomposed as $\EE(\sigma) = \EE_{\text{BS}}(\sigma) + \mathcal{\tilde E}(\sigma)$. We will make use of the well known integral inequality 
\begin{align}
	\left ( \int_0^T |f(u) | \dd u \right )^p \leq T^{p-1} \int_0^T |f(u)|^p \dd u, \quad p \geq 1. \label{intinequality}
\end{align} 
For the term $\EE_{\text{BS}}(\sigma)$, we have 
\begin{align*}
	|\EE_{\text{BS}}(\sigma)| \leq \sum_{|\alpha| = 3} C_{\alpha} M_{\alpha} (T, K)  S_0^{\alpha_1} |\xi_T - 1|^{\alpha_1} T^{ \alpha_2-1} \left ( \int_0^T (1 - \rho_u^2)^{\alpha_2}  |\sigma_u - \Ex(\sigma_u) |^{\alpha_2} \right ), 
\end{align*}
where we have used the integral inequality \cref{intinequality}. Applying the Cauchy-Schwarz inequality, we obtain
\begin{align*}
	\Ex |\EE_{\text{BS}}(\sigma)| &\leq \sum_{|\alpha| = 3} C_{\alpha} M_{\alpha} (T, K)  S_0^{\alpha_1}  T^{ \alpha_2-1} \left (\Ex |\xi_T - 1|^{2\alpha_1}\right )^{1/2} \left \{ \Ex \left ( \int_0^T (1 - \rho_u^2)^{\alpha_2}  |\sigma_u - \Ex(\sigma_u) |^{\alpha_2} \right )^2 \right \}^{1/2} \\
	&\leq \sum_{|\alpha| = 3} C_{\alpha} M_{\alpha} (T, K)  S_0^{\alpha_1}  T^{ \alpha_2-1} \left (\Ex |\xi_T - 1|^{2\alpha_1}\right )^{1/2} T^{1/2} \\ &\cdot \left \{ \Ex \left ( \int_0^T (1 - \rho_u^2)^{2 \alpha_2}  |\sigma_u - \Ex(\sigma_u) |^{2 \alpha_2} \right ) \right \}^{1/2}, 
\end{align*}
where we have used the integral inequality \cref{intinequality} for the second inequality. \\[.5cm]
For the term $\mathcal{\tilde E}(\sigma)$, notice that 
\begin{align*}
 \Ex (\mathcal{\tilde E} (\sigma)) &= \Ex_{\Qro_2} \left (\mathcal{\tilde E}(\sigma)  \xi_T^{-2} e^{\int_0^T \rho_u^2 \sigma_u \dd u }\ \right ) = \frac{1}{4} \partial_{xx} \text{Put}_{\text{BS}}(\hat x , \hat y) S_0^2 \Ex_{\Qro_2} \Bigg \{ \left ( \int_0^T \rho_u^2 (\sigma_u - \Ex_{\Qro_2}(\sigma_u) )\dd u \right )^3 \\ &\cdot \int_0^1 (1 - u)^2 e^{\int_0^T \rho_m^2 \Ex_{\Qro_2}(\sigma_m) \dd m } e^{u \int_0^T \rho_m^2 (\sigma_m  - \Ex_{\Qro_2}(\sigma_m)) \dd m } \dd u \Bigg \}. 
\end{align*}
Now for $u \in (0, 1)$, $e^{u \int_0^T \rho_m^2 (\sigma_m  - \Ex_{\Qro_2}(\sigma_m)) \dd m } \leq e^{u \int_0^T \rho_m^2 |\sigma_m  - \Ex_{\Qro_2}(\sigma_m)| \dd m }$. Thus
\begin{align*}
	\sup_{u \in (0,1) } e^{u \int_0^T \rho_m^2 (\sigma_m  - \Ex_{\Qro_2}(\sigma_m)) \dd m } \leq  e^{\ \int_0^T \rho_m^2 |\sigma_m  - \Ex_{\Qro_2}(\sigma_m)| \dd m }.
\end{align*}
Hence 
\begin{align*}
	\int_0^1 (1 - u)^2 e^{\int_0^T \rho_m^2 \Ex_{\Qro_2}(\sigma_m) \dd m } e^{u \int_0^T \rho_m^2 (\sigma_m  - \Ex_{\Qro_2}(\sigma_m)) \dd m } \dd u \leq \frac{1}{3} e^{\int_0^T \rho_m^2 \left ( \Ex_{\Qro_2}(\sigma_m) + |\sigma_m - \Ex_{\Qro_2}(\sigma_m) | \right )\dd m }. 
\end{align*} 
Thus 
\begin{align*}
 	\Ex |\mathcal{\tilde E} (\sigma)| &\leq C \partial_{xx} \text{Put}_{\text{BS}}(\hat x , \hat y) S_0^2  e^{\int_0^T \rho_m^2 \Ex_{\Qro_2}(\sigma_m)\dd m} \\& \cdot\Ex_{\Qro_2} \Bigg \{ \left ( \int_0^T \rho_u^2 (\sigma_u - \Ex_{\Qro_2}(\sigma_u) )\dd u \right )^3 e^{\int_0^T \rho_m^2  |\sigma_m - \Ex_{\Qro_2}(\sigma_m) |\dd m } \Bigg \}. 
\end{align*}
Finally, using the Cauchy-Schwarz inequality and the integral inequality \cref{intinequality}, we obtain 
\begin{align*}
 	\Ex |\mathcal{\tilde E} (\sigma)| &\leq C \partial_{xx} \text{Put}_{\text{BS}}(\hat x , \hat y) S_0^2  e^{\int_0^T \rho_m^2 \Ex_{\Qro_2}(\sigma_m)\dd m}T^{5/2} \Bigg ( \int_0^T \rho_u^{12} \Ex_{\Qro_2} |\sigma_u - \Ex_{\Qro_2}(\sigma_u) |^6 \dd u \Bigg )^{1/2} \\&\cdot \left (  \Ex_{\Qro_2} \left (e^{\int_0^T 2 \rho_m^2  |\sigma_m - \Ex_{\Qro_2}(\sigma_m) |\dd m } \right ) \right )^{1/2}. 
\end{align*}
Furthermore, notice that $ \partial_{xx} \text{Put}_{\text{BS}}(\hat x , \hat y) = \tilde M(T, K)$, where $\tilde M(T, K)$ is a function which behaves like $M(T, K)$.
\end{proof}

%--------Corr: Error bounds for general sigma with rho = 0
\begin{corollary}[Error bounds for general $\sigma$: $\rho = 0$]
\label{cor:errboundrho0}
For $\rho = 0$ a.e., the error term in the pricing formula is bounded as 
\begin{align*}
	\left | \Ex \left (\mathcal{E}_0(\sigma) \right ) \right | &\leq C M(T, K)T^2 \int_0^T \Ex | \sigma_t - \Ex (\sigma_t) |^3 \dd t,
\end{align*}
where $C = 1/6$ is a constant.
\end{corollary}
\begin{proof}
From \Cref{totalerrorrho0part2}, we have
\begin{align*}
\mathcal{E}_0(\sigma) &= \frac{1}{2}  \left ( \int_0^T (\sigma_t - \Ex(\sigma_t)) \dd t \right )^3  \int_0^1 (1 - u)^2 \partial_{yyy} \text{Put}_{\text{BS}} \left (S_0,   \tilde G(u) \right ) \dd u.
\end{align*} 
Notice that $( \int_0^T| f(t)| \dd t )^3 \leq T^2 \int_0^T |f(t)|^3 \dd t$ and use \Cref{parboundcor1}. Then the result is immediate.
\end{proof}

%--------Remark: What we have done
\begin{remark}
\Cref{thm:errbound} and \Cref{cor:errboundrho0} present bounds on the error generated by the expansion procedure in terms of higher-order moments of the variance process. Note that there are two components in the expansion procedure that mainly contribute to the error, the third-order partial derivatives of $\text{Put}_{\text{BS}}$ and the higher-order moments of the variance process. This is a clear consequence of how the remainder term behaves in Taylor's theorem; namely, the remainder is governed by the chosen function and the difference in the expansion and evaluation point. \Cref{parboundprop1} and \Cref{parboundcor1} demonstrate that $\frac{\partial^3 \text{Put}_{\text{BS}}}{\partial x^{\alpha_1} \partial y^{\alpha_2}}(F(u), G(u))$ for $u \in (0,1)$ are dominated by functions such as $M(T, K)$, which itself are bounded on $\reals^2_+$. The behaviour of $M(T,K)$ for a fixed $K$ and fixed $T$ is governed by the functions $\zeta(T)$ and $\eta(K)$ respectively. Interestingly, the error bound as a function of strike $K$ is purely governed by the function $\eta(K)$, and thus tends to zero for very small or very large strikes. However, the total behaviour in $T$ of the error bound depends not just on $\zeta(T)$, but on other terms as well, most notably the moments of the underlying variance process.

Indeed, in order to go a step further, one can consider a specific example of a variance process, and then attempt to bound the moments that appear in \Cref{thm:errbound} and \Cref{cor:errboundrho0}. However, for the purposes of application, it is perhaps more beneficial to instead consider a numerical sensitivity analysis in order to investigate the behaviour of the error with respect to parameter adjustments. Indeed, we have performed this in \Cref{sec:numerical2} for the Heston and GARCH diffusion models.
\end{remark}

%----------------------------------------------------------------------------Exp formulas and fast calibration-----------------------------------------------------------
%-----------------------------------------------------------------------------------------------------------------------------------------------------------------------------
%----------------------------------------------------------------------------Exp formulas and fast calibration-----------------------------------------------------------
%----------------------------------------------------------------------------------------------------------------------------------------------------------------------------------
%\clearpage
 \section{Closed-form formulas and fast calibration}
 \noindent
 \label{sec:fastcal2}
In this section, under the assumption of piecewise-constant parameters, we provide closed-form formulas for the price of a put option as well as a fast calibration scheme. To do this, we appeal to the results from \Cref{sec:pricingspecmodels2} on the approximation of prices of put options in the various models. We recognise that these prices are in terms of specific iterated integrals. Our goal is to show that these iterated integrals obey a convenient recursive property when parameters are piecewise-constant. To this end, define the integral operator 
\begin{align}
	\omega_T^{(k, l)} := \int_0^T l_u e^{\int_0^uk_z \dd z } \dd u. \label{int}
\end{align}
In addition, we define the $n$-fold integral operator using the following recurrence:
\begin{align}
	\omega_T^{(k^{(n)}, l^{(n)}), (k^{(n-1)}, l^{(n-1)}), \dots , (k^{(1)}, l^{(1)})}  := \omega_T^{\big (k^{(n)}, l^{(n)} w_\cdot ^{(k^{(n-1)}, l^{(n-1)}) ,\cdots, (k^{(1)}, 	l^{(1)})}\big ) }, \quad n \in \mathbb{N}. \label{intit}
\end{align}
%\footnote{For example 
%\begin{align*}
%\omega_T^{( k^{(3)}, l^{(3)}), (k^{(2)}, l^{(2)}), (k^{(1)}, l^{(1)} ) } = \int_0^T l^{(3)}_{u_3} e^{\int_0^{u_3} k^{(3)}_z \dd z } \left (  \int_0^{u_3}  l^{(2)}_{u_2} e^{\int_0^{u_2} k^{(2)}_z \dd z } \left (  \int_0^{u_2}  l^{(1)}_{u_1} e^{\int_0^{u_1} k^{(1)}_z \dd z } \dd u_1 \right ) \dd u_2 \right ) \dd u_3. 
%\end{align*}
%}
The rest of the section is dedicated to making these integral operators closed-form when parameters are piecewise-constant.

Let $ \mathcal{T} = \{ 0 = T_0, T_1, \dots, T_{N-1}, T_N = T \}$, where $T_i < T_{i+1}$ be a collection of maturity dates on $[0,T]$, with $\Delta T_i := T_{i+1} - T_i$ and $\Delta T_0 \equiv 1$. When the dummy functions are piecewise-constant, that is, $l^{(n)}_t = l^{(n)}_i$ on $ t \in [T_{i}, T_{i+1})$ and similarly for $k^{(n)}$, then we can recursively calculate the integral operators \cref{int} and \cref{intit}. Define 
\begin{align*}
e_t^{(k^{(n)}, \dots, k^{(1)})} &:= e^{\int_0^t \sum^n_{j=1} k_z^{(j)} \dd z}, \\
\varphi_{T_i, t}^{(k,p)} &:= \int_{T_i}^t \gamma_i^p(u) e^{\int_{T_i}^u k_z \dd z } \dd u,
\end{align*}
where $\gamma_i(u) :=( u - T_i)/\Delta T_i$ and $p \in \mathbb{N} \cup \{ 0 \}$. In addition, define the $n$-fold extension of $\varphi_{T_i, \cdot}^{(\cdot, \cdot)}$ recursively by
\begin{align*}
	\varphi_{{T_i}, t}^{(k^{(n)}, p_n), \dots, (k^{(2)}, p_2),( k^{(1)}, p_1)} := \int_{T_i}^t \gamma_i^{p_n} (u) e^{\int_{T_i}^u k_z^{(n)} \dd z } \varphi_{{T_i}, u}		^{(k^{(n-1)},p_{n-1}), \dots, (k^{(2)}, p_2), ( k^{(1)}, p_1)} \dd u,
\end{align*}
where $p_n \in \mathbb{N} \cup \{0 \}$.
%\footnote{For example 
%\begin{align*}
%\varphi_{{T_i}, t}^{(k^{(3)}, p_3), ( k^{(2)}, p_2) ,( k^{(1)},p_1)} =  \int_{T_i}^t \gamma_i^{p_3}(u_3)e^{\int_{T_i}^{u_3} k^{(3)}_z \dd z } \left (\int_{T_i}^{u_3}  \gamma_i^{p_2}(u_2)e^{\int_{T_i}^{u_2} k^{(2)}_z \dd z } \left ( \int_{T_i}^{u_2}  \gamma_i^{p_1}(u_1) e^{\int_{T_i}^{u_1} k^{(1)}_z \dd z } \dd u_1 \right ) \dd u_2 \right ) \dd u_3.
%\end{align*}}
With the assumption that the dummy functions are piecewise-constant, we can obtain the integral operator at time $T_{i+1}$ expressed by terms at time $T_i$.

%----------------To the editor, please put \displaybreak after a comma and just before the %---n fold comment in order to format this equation. This way, it will break up the equation but will keep each = portion on the same page.
\begin{align*} 
%---1 fold
\omega_{T_{i+1}}^{(k^{(1)}, l^{(1)})} &= \omega_{T_{i}}^{(k^{(1)}, l^{(1)})} + l_i^{(1)} e_{T_i}^{(k^{(1)})} \varphi_{T_i, T_{i+1}}^{(k^{(1)},0)}, \\
%---2 fold
\omega_{T_{i+1}}^{(k^{(2)}, l^{(2)}), (k^{(1)}, l^{(1)})} &= \omega_{T_{i}}^{(k^{(2)}, l^{(2)}), ( k^{(1)}, l^{(1)})}\\ &+  l_i^{(2)} e_{T_i}^{(k^{(2)})} \varphi_{T_i, T_{i+1}} ^{(k^{(2)},0)} \omega_{T_{i}}^{(k^{(1)}, l^{(1)})} \\&+ l_i^{(2)} l_i^{(1)} e_{T_i}^{(k^{(2)}, k^{(1)})} \varphi_{T_i, T_{i+1}}^{(k^{(2)}, 0) ,( k^{(1)},0)}, \\
%---3 fold
\omega_{T_{i+1}}^{(k^{(3)}, l^{(3)}),(k^{(2)}, l^{(2)}), (k^{(1)}, l^{(1)})} &=  \omega_{T_{i}}^{(k^{(3)}, l^{(3)}),(k^{(2)}, l^{(2)}), (k^{(1)}, l^{(1)})} \\&+ l_i^{(3)}  e_{T_i}^{(k^{(3)})} \varphi_{T_i, T_{i+1}}^{(k^{(3)}, 0)} \omega_{T_i} ^{(k^{(2)}, l^{(2)}), (k^{(1)}, l^{(1)})} \\
&+ l_i^{(3)} l_i^{(2)} e_{T_i}^{(k^{(3)}, k^{(2)})} \varphi_{T_i, T_{i+1}}^{(k^{(3)},0), ( k^{(2)},0)} \omega_{T_i}^{(k^{(1)}, l^{(1)})} \\&+ l_i^{(3)} l_i^{(2)} l_i^{(1)} e_{T_i}^{(k^{(3)}, k^{(2)}, k^{(1)})} \varphi_{T_i, T_{i+1}}^{(k^{(3)},0), ( k^{(2)}, 0), ( k^{(1)}, 0)}, \\
%---4 fold
\omega_{T_{i+1}}^{(k^{(4)}, l^{(4)}), \dots , (k^{(1)}, l^{(1)})} &=  \omega_{T_{i}}^{(k^{(4)}, l^{(4)}), (k^{(3)}, l^{(3)}),(k^{(2)}, l^{(2)}), (k^{(1)}, l^{(1)})} \\&+ l_i^{(4)}  e_{T_i}^{(k^{(4)})} \varphi_{T_i, T_{i+1}}^{(k^{(4)},0)} \omega_{T_i} ^{(k^{(3)}, l^{(3)}), (k^{(2)}, l^{(2)}), (k^{(1)}, l^{(1)})} \\
&+ l_i^{(4)} l_i^{(3)} e_{T_i}^{(k^{(4)}, k^{(3)})} \varphi_{T_i, T_{i+1}}^{(k^{(4)},0), ( k^{(3)}, 0)} \omega_{T_i}^{(k^{(2)}, l^{(2)}),(k^{(1)}, l^{(1)})} \\&+ l_i^{(4)} l_i^{(3)} l_i^{(2)} e_{T_i}^{(k^{(4)}, k^{(3)}, k^{(2)})} \varphi_{T_i, T_{i+1}}^{(k^{(4)}, 0), (k^{(3)}, 0), ( k^{(2)}, 0)} \omega_{T_i}^{(k^{(1)}, l^{(1)})} \\&+ l_i^{(4)} l_i^{(3)} l_i^{(2)} l_i^{(1)}e_{T_i}^{(k^{(4)}, k^{(3)}, k^{(2)}, k^{(1)})} \varphi_{T_i, T_{i+1}}^{(k^{(4)}, 0),(k^{(3)}, 0), (k^{(2)}, 0),( k^{(1)}, 0)}, 		\iffalse	 !!!!For example, you could put a \displaybreak here!!!!		 \fi \\
%---5 fold
\omega_{T_{i+1}}^{(k^{(5)}, l^{(5)}),\dots , (k^{(1)}, l^{(1)})} &=  \omega_{T_{i}}^{(k^{(5)}, l^{(5)}),(k^{(4)}, l^{(4)}), (k^{(3)}, l^{(3)}),(k^{(2)}, l^{(2)}), (k^{(1)}, l^{(1)})} \\&+ l_i^{(5)}  e_{T_i}^{(k^{(5)})} \varphi_{T_i, T_{i+1}}^{(k^{(5)},0)} \omega_{T_i} ^{(k^{(4)}, l^{(4)}),(k^{(3)}, l^{(3)}), (k^{(2)}, l^{(2)}), (k^{(1)}, l^{(1)})} \\
&+ l_i^{(5)} l_i^{(4)} e_{T_i}^{(k^{(5)}, k^{(4)})} \varphi_{T_i, T_{i+1}}^{(k^{(5)},0), ( k^{(4)}, 0)} \omega_{T_i}^{(k^{(3)}, l^{(3)}),(k^{(2)}, l^{(2)}),(k^{(1)}, l^{(1)})} \\&+ l_i^{(5)} l_i^{(4)} l_i^{(3)} e_{T_i}^{(k^{(5)}, k^{(4)}, k^{(3)})} \varphi_{T_i, T_{i+1}}^{(k^{(5)}, 0), (k^{(4)}, 0), ( k^{(3)}, 0)} \omega_{T_i}^{(k^{(2)}, l^{(2)}),(k^{(1)}, l^{(1)})} \\&+ l_i^{(5)} l_i^{(4)} l_i^{(3)} l_i^{(2)}e_{T_i}^{(k^{(5)}, k^{(4)}, k^{(3)}, k^{(2)})} \varphi_{T_i, T_{i+1}}^{(k^{(5)}, 0),(k^{(4)}, 0), (k^{(3)}, 0),( k^{(2)}, 0)}\omega_{T_i}^{(k^{(1)}, l^{(1)})} \\
&+l_i^{(5)} l_i^{(4)} l_i^{(3)} l_i^{(2)} l_i^{(1)}e_{T_i}^{(k^{(5)}, k^{(4)}, k^{(3)}, k^{(2)}, k^{(1)})} \varphi_{T_i, T_{i+1}}^{(k^{(5)}, 0),(k^{(4)}, 0), (k^{(3)}, 0),( k^{(2)}, 0), (k^{(1)}, 0)}.
\end{align*}

%\footnote{In general
%\begin{align*}
%	\omega_{T_{i+1}}^{(k^{(n)}, l^{(n)}), \dots , (k^{(2)}, l^{(2)}), (k^{(1)}, l^{(1)})} &= \sum_{m = 1}^{n+1} \omega_{T_i}^{(k^{(n-m+1)}, l^{(n-m+1)}), \dots , (k^{(1)}, l^{(1)})} \left (\prod_{j=0}^{m-2} l_i^{(n-j)} \right )e_{T_i}^{(k^{(n-m+2)}, \dots, k^{(1)})} \\ &\cdot\varphi_{T_i, T_{i+1}}^{(k^{(n-m+2)},0),  \dots, (k^{(1)}, 0)},
%\end{align*}
%where whenever the index goes outside of $\{1, \dots, n \}$, then that term is equal to $1$.
%}

The only terms here that are not explicit are the functions $e_\cdot^{(\cdot, \dots , \cdot)}$ and $\varphi_{T_i, \cdot}^{(\cdot, \cdot), \dots , (\cdot, \cdot)}$. For $t \in (T_i, T_{i+1}]$, we can derive the following: 
\begin{align*}
e_t^{(k^{(n)}, \dots , k^{(1)})} = 
	e_{T_i}^{(k^{(n)}, \dots , k^{(1)})} e^{ \Delta T_i \gamma_i(t) \sum_{j=1}^n k_i^{(j)}  }= e^{\sum_{m=0}^{i-1}\Delta T_m \sum_{j=1}^n k_m^{(j)} } e^{ \Delta 		T_i \gamma_i(t) \sum_{j=1}^n k_i^{(j)}},
\end{align*}
where $e_0^{(k^{(n)}, \dots, k^{(1)})} = 1$. Using integration by parts and basic integration properties, we obtain
\begin{align*}
 \varphi_{T_i, t}^{(k, p)} &= 
 	\begin{cases} \frac{1}{k_i} \left ( \gamma_i^p(t) e^{k_i \Delta T_i \gamma_i(t) } - \frac{p}{\Delta T_i} \varphi_{T_i, t}^{(k, p-1)} \right ), \quad &k_i \neq 0, p \geq 1, \\[.2cm] 
 	\frac{1}{k_i} \left ( e^{k_i \Delta T_i \gamma_i(t) } - 1 \right ), \quad &k_i \neq 0, p =0, \\[.2cm] 
 	\frac{1}{p+1} \Delta T_i \gamma_i^{p+1}(t),  &k_i = 0, p \geq 0.
 	\end{cases}
\end{align*}
In addition, for $n \geq 2$,
\begin{align*}
\varphi_{{T_i}, t}^{(k^{(n)}, p_n),  \dots, (k^{(1)}, p_1)} &= 
	\begin{cases} \frac{1}{k_i^{(n)}} \Big (\gamma_i^{p_n}(t) e^{ k_i^{(n)} \Delta T_i \gamma_i(t)} \varphi_{{T_i}, t}^{(k^{(n-1)}, p_{n-1}), \dots,( k^{(1)},p_1)} \\-  \frac{p_n}{\Delta T_i} \varphi_{{T_i}, t}^{(k^{(n)}, p_n-1), (k^{(n-1)}, p_{n-1}), \dots, (k^{(1)}, p_1)}  \\-  \varphi_{{T_i}, t}^{(k^{(n)} + k^{(n-1)}, p_n + p_{n-1} ),  (k^{(n-2)}, p_{n-2} ), \dots, (k^{(1)}, p_1)} \Big ), 		&k_i^{(n)} \neq 0, p_n \geq 1, \\[.3cm]
	\frac{1}{k_i^{(n)}} \Big (e^{ k_i^{(n)} \Delta T_i \gamma_i(t)} \varphi_{{T_i}, t}^{(k^{(n-1)}, p_{n-1}), \dots,( k^{(1)},p_1)}  \\-  \varphi_{{T_i}, t}^{(k^{(n)} + k^{(n-1)},  p_{n-1} ),  (k^{(n-2)}, p_{n-2} ), \dots, (k^{(1)}, p_1)} \Big ), 		&k_i^{(n)} \neq 0, p_n =0, \\[.3cm]
	\frac{\Delta T_i }{p_n + 1} \Big (  \gamma_i^{p_n +1 } (t)  \varphi_{{T_i}, t}^{(k^{(n-1)}, p_{n-1}), \dots,( k^{(1)}, p_1)} \\- \varphi_{T_i, t} ^{(k^{(n-1)}, p_n + p_{n-1} + 1 ), (k^{(n-2)}, p_{n-2}), \dots , (k^{(1)}, p_1)} \Big ), &k_i^{(n)} = 0, p_n \geq 0.
\end{cases}
\end{align*}

%-----------------------------------------------------------------------------Explicit prices---------------------------------------------------------------------------------
\subsection{Closed-form prices}
We now express the second-order put option prices under the Heston and GARCH diffusion models in terms of the integral operators \cref{int} and \cref{intit}. These are given by the following propositions. By doing so, we essentially obtain a closed-form expression for the put option prices when parameters are piecewise-constant. This is a consequence of the relationships derived in this section. Since we are essentially expressing the results from \Cref{sec:pricingspecmodels2} with new notation, we omit the proofs for these propositions.

%--------Proposition: Heston explicit price
\begin{proposition}[Second-order Heston explicit put option price]
\label{prop:hestonCFA}
The second-order put option price in the Heston model can be written as
\begin{align*}
\text{Put}^{(2)}_{\text{H}} &=  \text{Put}_{\text{BS}} ( \hat x, \hat y) \\ &+ \frac{1}{2} \partial_{xx} \text{Put}_{\text{BS}} ( \hat x, \hat y)  S_0^2 \Ex (\xi _T - 1)^2 + \frac{1}{2}\partial_{yy} \text{Put}_{\text{BS}} ( \hat x, \hat y)   \Ex \left (\int_0^T (1 - \rho_t^2) (V_t - \Ex(V_t)) \dd t \right )^2 \\ 
&+\partial_{xy} \text{Put}_{\text{BS}} ( \hat x, \hat y) S_0 \Ex \left \{  ( \xi_T - 1) \left (\int_0^T (1 - \rho_t^2) (V_t - \Ex(V_t)) \dd t \right ) \right \},
\end{align*}
where
\begin{align*}
\Ex(\xi_T - 1 )^2 &\approx \left (\exp \Bigg \{v_0 \omega_T^{(- (\kappa - 2 \lambda \rho), \rho^2 )} + \omega_T^{(-(\kappa - 2\lambda \rho), \rho^2), (\kappa - 2\lambda \rho, \kappa \theta)} \Bigg \} \right ) \\
&\cdot \Bigg \{1+v_0 \omega_T^{(-(\kappa - 2 \lambda \rho), \rho^2), (-(\kappa - 2\lambda \rho), \rho^2), (\kappa - 2 \lambda \rho, \lambda^2 )} \\ &+ \omega_T ^{( - (\kappa - 2\lambda \rho ) , \rho^2 ), (-(\kappa - 2 \lambda \rho), \rho^2), (\kappa - 2 \lambda \rho, \lambda^2), (\kappa - 2 \lambda \rho , \kappa \theta)} \Bigg \} - 1. 
\end{align*}
\begin{align*}
\Ex \left (\int_0^T (1 - \rho_t^2) (V_t - \Ex(V_t)) \dd t \right )^2 &= 2 v_0 \omega_T^{(- \kappa, 1 - \rho^2), (-\kappa, 1 - \rho^2), ( \kappa, \lambda^2)} \\&+ 2 \omega_T^{(-\kappa, 1 - \rho^2), (-\kappa, 1 - \rho^2), ( \kappa, \lambda^2), ( \kappa, \kappa \theta) }, \\
\Ex \left \{  ( \xi_T - 1) \left (\int_0^T (1 - \rho_t^2) (V_t - \Ex(V_t)) \dd t \right ) \right \} &=  v_0 \left ( \omega_T^{(- ( \kappa - \lambda \rho ), 1 - \rho^2)} - \omega_T^{(- \kappa , 1 - \rho^2)}\right ) \\&+ \omega_T^{(-(\kappa - \lambda \rho), 1- \rho^2), (\kappa - \lambda \rho, \kappa \theta)}  \\&-  \omega_T^{(-\kappa, 1- \rho^2), (\kappa, \kappa \theta)}. 
\end{align*}
Furthermore $\hat x = S_0$ and $\hat y = v_0 \omega_T^{(- \kappa, 1 - \rho^2)} + \omega_T^{(-\kappa, 1- \rho^2), ( \kappa, \kappa \theta)}.$
\end{proposition}

%----------------------------------------------------------------------------GARCH calibration-------------------------------------------------------------------------------
%--------Proposition: GARCH explicit price
\begin{proposition}[Second-order GARCH explicit put option price]
\label{prop:garchCFA}
The second-order put option price in the GARCH diffusion model can be written as 
\begin{align*}
\text{Put}^{(2)}_{\text{GARCH}} &=  \text{Put}_{\text{BS}} ( \hat x, \hat y) + \partial_{yy} \text{Put}_{\text{BS}} ( \hat x, \hat y) \int_0^T\left ( \int_0^t \Cov(V_s, V_t) \dd s \right ) \dd t,
\end{align*}
where
\begin{align*}
  \int_0^T\left ( \int_0^t \Cov(V_s, V_t) \dd s \right ) \dd t &= \Big ( v_0^2 \omega_T^{(-\kappa, 1),  (-\kappa, 1) ,(\lambda^2, \lambda^2)}  + 2v_0 		\omega_T^{(-\kappa, 1), (-\kappa, 1),  (\lambda^2, \lambda^2), (-(\lambda^2 - \kappa ), \kappa \theta)}  \\
	&+ 2 \omega_T^{(-\kappa, 1),(-\kappa,1), (\lambda^2, \lambda^2 ), (-(\lambda^2 - \kappa),\kappa \theta ), (\kappa,\kappa \theta ) } \Big ).
\end{align*}
Furthermore $\hat x = S_0$ and 
\begin{align*}
	\hat y = \int_0^T \Ex(V_t) \dd t = v_0 \omega_T^{(-\kappa, 1 )} +  \omega_T^{(-\kappa, 1), (\kappa, \kappa \theta)}.
\end{align*}
\end{proposition}
\begin{remark}[Fast calibration]
 Since the integral operators obey a recursion, we can exploit dynamic programming to greatly speed up the calibration procedure.
To implement our fast calibration scheme, one executes the following algorithm. Let $\mu_t \equiv \mu = (\mu^{(1)}, \mu^{(2)}, \dots, \mu^{(n)})$ be an arbitrary set of parameters and denote by $\omega_t$ an arbitrary integral operator.
\begin{itemize}
\item Calibrate $\mu$ over $[0,T_1)$ to obtain $\mu_0$. This involves computing $\omega_{T_1}$.
\item Calibrate $\mu$ over $[T_1, T_2)$ to obtain $\mu_1$. This involves computing $\omega_{T_2}$ which is in terms of $\omega_{T_1}$, the latter already being computed in the previous step.
\item Repeat until time $T_N$.
\end{itemize}
\end{remark}

%----------------------------------------------------------------Numerical tests and sensitivity analysis------------------------------------------------------------------
%------------------------------------------------------------------------------------------------------------------------------------------------------------------------------
%-----------------------------------------------------------------Numerical tests and sensitivity analysis--------------------------------------------------------------------
%---------------------------------------------------------------------------------------------------------------------------------------------------------------------------------
\section{Numerical tests and sensitivity analysis}

\newcommand{\wo}{\textcolor{white}{0}}	% Use this guy to put an extra space if number of significant figures in column is not consistent
\newcommand\red[1]{\color{red}#1}         %\red gives red text

\noindent
\label{sec:numerical2}
In this section, we will numerically investigate the accuracy of our closed-form approximation formula in the Heston and GARCH diffusion models (\Cref{prop:hestonCFA} and \Cref{prop:garchCFA} respectively) by considering the sensitivity of our approximation formula whilst varying one parameter at a time, with the others all being fixed. Namely, for an arbitrary set of piecewise-constant parameters $(\kappa_t, \theta_t, \lambda_t , \rho_t) \equiv (\kappa, \theta, \lambda, \rho) =:\mu $, we vary only one of $\kappa$, $\theta$, $\lambda$, $\rho$ at a time and keep the rest fixed. Then, we will compute the difference (signed error) in implied volatilities via our approximation formula as well as the Monte-Carlo for maturity times $T \in \{1/12, 3/12, 6/12, 1\} \equiv \{1\text{M}, 3\text{M}, 6\text{M}, 1\text{Y}\}$ and strikes corresponding to Put 10, 25, and ATM deltas. Thus, the signed error of the implied volatility for a given parameter set $\mu$, maturity $T$, and strike $K$ is
\begin{align*}
	\text{Error}(\mu, T, K) = \sigma_{\text{IM-Approx}}(\mu, T, K)  - \sigma_{\text{IM-Monte}}(\mu, T, K).
\end{align*}

The Monte-Carlo simulation is implemented by taking advantage of a variation of the mixing solution methodology from \Cref{appen:mixing}. Namely, consider the log-spot $X_t = \ln S_t$ and log-strike $k = \ln K$. Let $\sigma$ be the variance process. Then
\begin{align*}
	\text{Put} &= e^{-\int_0^T r_t^d \dd t } \Ex  ( e^k- e^{X_T})_+  \\
			&= \Ex \left (  e^{-\int_0^T r_t^d \dd t } \Ex \big  [ (e^k - e^{X_T})_+ | \FF_T^B \big ] \right ) \\
			&= \Ex \left [ P_{\text{BS}} \left ( x_0 - \int_0^T \frac{1}{2} \rho_t ^2 \sigma_t \dd t + \int_0^T \rho_t \sqrt{\sigma_t}  \dd B_t , \int_0^T \sigma_t ( 1 - \rho_t^2) \dd t \right ) \right ],
\end{align*}
where 
%------------Defn: P_BS
\begin{align}
\begin{split}
	P_{\text{BS}}(x, y) &:= e^k e^{-\int_0^T r_t^d \dd t }  \NN ( -d^{\ln}_-) -   e^x  e^{-\int_0^T r_t^f \dd t } \NN (-d^{\ln}_+ ),  \\
	d^{\ln}_{\pm}:= d^{\ln}_{\pm}(x,y) &:= \frac{ x - k + \int_0^T ( r_t^d - r_t^f ) \dd t}{\sqrt{y}} \pm \frac{1}{2} \sqrt{y}.
\end{split}
\end{align}
We will not provide a proof of this result, as it can be derived by adapting the proof of \Cref{thm:mixing} to the case of modelling the log-spot in a straightforward manner. Using this relationship for Monte-Carlo simulation, there is no need to simulate $S$, one need only simulate $\sigma$. This reduces the runtime as well as standard error of the procedure.

We will employ the usual Euler-Maruyama method to simulate the variance process. For all our simulations, we use 2,000,000 Monte Carlo paths, and 24 time steps per day, where a year is comprised of 252 trading days. This is to reduce the Monte-Carlo and discretisation errors sufficiently well.

%-----------Remark: Code is provided.
\begin{remark}
The code utilised to obtain the numerical results in this section is available at GitHub \citep{das2021}. In particular, what is provided is: 
\begin{itemize}
\item A routine which computes our closed-form approximation of put option prices for both the Heston model and GARCH diffusion model (the GARCH diffusion model with $\rho = 0$ a.e.) with piecewise-constant parameter inputs.
\item A routine which implements the Monte-Carlo simulation via the mixing solution methodology for the pricing of put option prices in the Heston model and GARCH diffusion model with piecewise-constant parameter inputs.
\item A routine which compares the accuracy and runtimes of the aforementioned methods.
\end{itemize}
\end{remark}

%-------------------------------------------------------------------------Heston sensitivity analysis---------------------------------------------------------------------------
\subsection{Heston sensitivity analysis}

%-----------Remark: n piece parameter
\begin{remark}
A piecewise-constant parameter which is piecewise-constant over $n$ intervals will be called an $n$-piece parameter or a parameter with $n$ pieces.
\end{remark}

\label{hestonsense} We start from a ``safe'' parameter set, which are constant parameters calibrated by Bloomberg USD/JPY FX option price data \citep{bloombergdata}. The safe parameter set is 

%----Pseudo table: S0, v0, rd, rf
\begin{center}
\begin{tabular}{llll} 
\toprule
$S_0$ & \ $v_0$ & $r^d$ & $r^f$ \\
\midrule
$100$ & \ $0.36\%$ &  0.02 & 0 \\
 \bottomrule \\
\end{tabular}
\end{center}
with

%----Pseudo table: kap, the, lam, rho

%\begin{table}[H]
\begin{center}
\begin{tabular}{l c cccc}
\toprule

 $T$	 && 		$\kappa$ & $\theta$ & $\lambda$ & $\rho$   \\
\midrule
1M   && 		$5.00$   & 1.90\%   & 0.414  & -0.391  \\
  
3M   && 		$5.00$  	& 1.10\%   & 0.414  & -0.391 \\

6M   &&	 	$5.00$   & 0.90\%   & 0.414  & -0.391  \\

1Y   && 		$5.00$   & 0.90\%   & 0.414  & -0.391  \\
\bottomrule
\end{tabular}
\end{center}

%\end{table}

However, we would like to consider piecewise-constant parameters, as our closed-form approximation method is designed in order to take advantage of piecewise-constant parameter inputs. To do so, we will make our parameters piecewise-constant over three time intervals, the length of the first, second and third time interval having proportions $1/4$, $1/4$, $1/2$ of the maturity time $T$ respectively. For example, if $T = 1/12$ then a 3-piece parameter is piecewise-constant on the time intervals $[0, 1/48), [1/48, 2/48), [2/48, 4/48)$. To choose the ``safe'' values of the 3-piece parameters, we will simply (albeit somewhat artificially) perturb the value of each safe parameter over each time interval, yielding the following table of ``safe'' 3-piece parameter sets:

%----Pseudo table: S0, v0, rd, rf, piecewise constant
\begin{center}
\begin{tabular}{lcc c cccc cc}
\toprule
$T$	 & Piece & Proportion	&& 		$\kappa$ & $\theta$ & $\lambda$ & $\rho$ & 	$r^d$ & $r^f$   \\

\midrule
1M   		& 1 & 1/4 		&&	 	4.80   & 1.70\%  & 0.394 & -0.371 & 			1\% & 0 \\
   		& 2 & 1/4 	   	&&	 	5.20   & 2.10\%  & 0.434 & -0.411 &			3\% & 0 \\
   		& 3 & 1/2 		&&	 	5.00   & 1.90\%  & 0.414 & -0.391 & 			2\% & 0 \\

\midrule
3M   		& 1 & 1/4  	&&	 	4.80   & 0.90\%  & 0.394 & -0.371 &			1\% & 0   \\
   		& 2 & 1/4	   	&&	 	5.20   & 1.30\%  & 0.434 & -0.411 & 			3\% & 0   \\
   		& 3 & 1/2   	&&	 	5.00   & 1.10\%  & 0.414 & -0.391 & 			2\% & 0  \\

\midrule
6M   		& 1 & 1/4   	&&	 	4.80   & 0.70\%  & 0.394 & -0.371 & 			1\% & 0	\\
   		& 2 & 1/4	   	&&	 	5.20   & 1.10\%  & 0.434 & -0.411 & 			3\% & 0 	 \\
   		& 3 & 1/2 		&&	 	5.00   & 0.90\%  & 0.414 & -0.391 & 			2\% & 0 	 \\
    
\midrule
1Y   		& 1 & 1/4  	&&	 	4.80   & 0.70\%  & 0.394 & -0.371 &  			1\% & 0 	\\
   		& 2 & 1/4	   	&&	 	5.20   & 1.10\%  & 0.434 & -0.411 & 			3\% & 0	 \\
   		& 3 & 1/2   	&&	 	5.00   & 0.90\%  & 0.414 & -0.391 & 			2\% & 0	 \\

\bottomrule
\end{tabular}
\end{center}

In our numerical analysis, we vary one of the 3-piece parameters $(\kappa, \theta, \lambda, \rho)$ with the rest fixed, and then compute implied volatilities via both the closed-form approximation formula as well as the Monte-Carlo method as described above. Specifically, we select a 3-piece parameter from $(\kappa, \theta, \lambda, \rho)$, and start at 40\% of its safe 3-piece parameter value, then increase the value of each piece in increments of 20\%, all the way up to 160\%, whilst keeping the other three 3-piece parameters fixed at their safe value. We then repeat this process with each of the other three 3-piece parameters. The relevant tables are \Cref{table:varyingkappaheston}, \Cref{table:varyingthetaheston}, \Cref{table:varyinglambdaheston}, and \Cref{table:varyingrhoheston}, for the analysis of $\kappa$, $\theta$, $\lambda$, and $\rho$ respectively. Note that the $\rho$ table values are reversed, as $\rho$ is negative, and thus 160\% of the safe 3-piece $\rho$ is approximately $(-0.594, -0.658, -0.626)$. This ensures that the parameter values are increasing for all tables.

%----Remark: Feller condition
\begin{remark}
The Feller condition is 
\begin{align*}
	2\kappa \theta > \lambda^2.
\end{align*}
\Crefrange{table:varyingkappaheston}{table:varyingrhoheston} use red text to indicate when the Feller condition is not satisfied. Note that in application, this condition is almost always violated. That is, parameters calibrated from market data almost always violate the Feller condition, see for example \citep{clark2011foreign,da2011riding,ribeiro2013approximation}. In fact, this is a reason why practitioners now favour non-affine models.
\end{remark}

%------------------------------------------------------------------------------Heston tables-----------------------------------------------------------------------------------

%%---------Varying kappa in Heston model 
%\label{sec:varyingkappaheston}
%\wo

%---Table: Signed errors in the Heston model with kappa varying (Done and checked!)
\begin{table}[H]
\caption{Signed error of implied volatilities in basis points, computed in the Heston model with $\kappa$ varying from 40\% to 160\% of its ``safe'' 3-piece value.}
\label{table:varyingkappaheston}
\begin{center}
\begin{tabular}{ll lllllll} 
\toprule
	&	{$\kappa$} & \ 40\% \wo & \ 60\% \wo  & \ 80\% & \ 100\% & \ 120\% & \ 140\% & \ 160\% \\
\midrule

ATM 		& {1M}   &	\red{\ 0.77} \wo 	&\red{\ 1.08} \wo 	&\red{\ 1.23} \wo	&\ 1.32		&\  1.37		&\ 1.40 		&\ 1.41   \\

 		& {3M}  &\red{-26.35}		&\red{-15.43}		&\red{-9.00}		&\red{-4.96}	&\red{-2.30}	&\red{-0.45}	&\red{\ 0.87} \\
 	
 		& {6M}  &\red{-59.94}		&\red{-31.56}		&\red{-16.97}		&\red{-8.82}	&\red{-3.96}	&\red{-0.93}	&\red{\ 1.00} \\
 	
 		& {1Y}  &\red{-70.41}		&\red{-29.66}		&\red{-12.76}		&\red{-4.92}	&\red{-1.05}	&\red{\ 0.95}	&\red{\ 1.98} \\ 
 \midrule

 Put 25 	& {1M}  &\red{\ 12.81}		&\red{\ 11.01}		&\red{\ 9.65}		&\ 8.61		&\ 7.79 		&\ 7.13 		&\ 6.59  \\
 
		& {3M}  &\red{\ 14.05}		&\red{\ 12.89}		&\red{\ 12.08}		&\red{\ 11.43}	&\red{\ 10.89}	&\red{\ 10.41}	&\red{\ 9.98}  \\
	
 		& {6M}  &\red{\ 3.83}		&\red{\ 9.08}		&\red{\ 11.35}		&\red{\ 12.14}	&\red{\ 12.17}	&\red{\ 11.82}	&\red{\ 11.29}  \\
 	
 		& {1Y}  &\red{-7.47}			&\red{\ 6.88}		&\red{\ 11.14}		&\red{\ 11.86}	&\red{\ 11.34}	&\red{\ 10.42}	&\red{\ 9.44} \\
 \midrule

Put 10 	& {1M}  &\red{\ 34.21}		&\red{\ 27.41}		&\red{\ 22.52}		&\ 18.91		&\ 16.19 		&\ 14.08	 	&\ 12.40  \\

		& {3M}  &\red{\ 86.32}		&\red{\ 62.63}		&\red{\ 47.89}		&\red{\ 38.17}	&\red{\ 31.40}	&\red{\ 26.50}	&\red{\ 22.81} \\
	
		& {6M}  &\red{\ 115.94}		&\red{\ 78.44}		&\red{\ 57.53}		&\red{\ 44.60}	&\red{\ 35.91}	&\red{\ 29.68}	&\red{\ 25.04} \\
	
 		& {1Y}   &\red{\ 105.65}		&\red{\ 66.94}		&\red{\ 47.08}		&\red{\ 35.05}	&\red{\ 27.06}	&\red{ \ 21.45}	&\red{\ 17.38} \\
 \bottomrule
\end{tabular}
\end{center}
\end{table}

%--------Varying theta in Heston model 
\label{sec:varyingthetaheston}
\wo

%---Table: Signed errors in the Heston model with theta varying (Done and checked!)
\begin{table}[H]
\caption{Signed error of implied volatilities in basis points, computed in the Heston model with $\theta$ varying from 40\% to 160\% of its ``safe'' 3-piece value.}
\label{table:varyingthetaheston}
\begin{center}
\begin{tabular}{ll lllllll} 
\toprule
	&	{$\theta$} & \ 40\% & \ 60\%  & \ 80\%  & \ 100\% & \ 120\% & \ 140\% & \ 160\% \\
\midrule

 ATM	 	& {1M}  &\red{-0.12} \wo		&\red{ \ 0.22} \wo	&\red{ \ 0.32} \wo	&\ 0.32		&\ 0.25		&\ 0.14		&\ 0.01	\\
 
		& {3M}  &\red{-20.59}		&\red{-12.98}		&\red{-8.66}		&\red{-6.05}	&\red{-4.41}	&\red{-3.35}	&\red{-2.65} \\
		
 		& {6M}  &\red{-32.70}		&\red{-20.25}		&\red{-13.51}		&\red{-9.46}	&\red{-6.90}	&\red{-5.19}	&\red{-4.02} \\
 		
 		& {1Y}   &\red{-21.50}		&\red{-12.02}		&\red{-7.09}		&\red{-4.24}	&\red{-2.45}	&\red{-1.28}	&\red{-0.47} \\
 \midrule

Put 25 	& {1M}  &\red{\ 14.04}		&\red{\ 11.50}		&\red{\ 9.69}		&\ 8.35		&\ 7.31		&\ 6.47		&\ 5.77 \\

 		& {3M}  &\red{\ 24.69}		&\red{\ 17.26}		&\red{\ 13.10}		&\red{\ 10.49}	&\red{\ 8.68}	&\red{\ 7.33}	&\red{\ 6.26} \\
 		
 		& {6M}  &\red{\ 31.34}		&\red{\ 20.93}		&\red{\ 15.91}		&\red{\ 13.07}	&\red{\ 11.24}	&\red{\ 9.94}	&\red{\ 8.93} \\
 		
		& {1Y}   &\red{\ 28.59}		&\red{\ 19.26}		&\red{\ 15.15}		&\red{\ 12.80}	&\red{\ 11.30}	&\red{\ 10.23}	&\red{\ 9.41} \\
\midrule

Put 10	& {1M}  &\red{\ 33.55}		&\red{\ 26.93}		&\red{\ 22.07}		&\ 18.45		&\ 15.69		&\ 13.53		&\ 11.79 \\

 		& {3M}  &\red{\ 84.65}		&\red{\ 63.16}		&\red{\ 48.78}		&\red{\ 38.90}	&\red{\ 31.87}	&\red{\ 26.69}	&\red{\ 22.76} \\
 		
 		& {6M}  &\red{\ 103.04}		&\red{\ 74.71}		&\red{\ 56.46}		&\red{\ 44.40}	&\red{\ 36.11}	&\red{\ 30.16}	&\red{\ 25.72} \\
 		
 		& {1Y}   &\red{\ 81.69}		&\red{\ 58.22}		&\red{\ 43.82}		&\red{\ 34.74}	&\red{\ 28.72}	&\red{\ 24.48}	&\red{\ 21.36} \\
 \bottomrule
\end{tabular}
\end{center}
\end{table}

%%--------Varying lambda in Heston model 
\label{sec:varyinglambdaheston}
\wo

%---Table: Signed errors in the Heston model with lambda varying (Done and checked!)
\begin{table}[H]
\caption{Signed error of implied volatilities in basis points, computed in the Heston model with $\lambda$ varying from 40\% to 160\% of its ``safe'' 3-piece value.}
\label{table:varyinglambdaheston}
\begin{center}
\begin{tabular}{ll lllllll} 
\toprule
	&	{$\lambda$} & \ 40\% & \ 60\%  & \ 80\%  & \ 100\% & \ 120\% & \ 140\% & \ 160\% \\
\midrule

 ATM		& {1M}  &\ 0.04 \wo  		&\ 0.59 \wo	 &\ 0.78 \wo 		&\ 0.15 		&\red{-2.00}	&\red{-6.32} 	&\red{-13.33}	\\
 
		& {3M}  &-0.07  		&\ 0.15		& \red{-1.68}		&\red{-7.00}	&\red{-17.03}	&\red{-32.53}	&\red{-54.16} \\
		
 		& {6M}  &\ 0.98			&\ 0.74		&\red{-2.26}		&\red{-9.55}	&\red{-22.26}	&\red{-41.14}	&\red{-66.66} \\
 		
 		& {1Y}  &\ 2.37			&\ 2.28		&\red{\ 0.41}		&\red{-4.38}	&\red{-12.88}	&\red{-25.70}	&\red{-43.27} \\
 \midrule

 Put 25 	& {1M}  &\ 1.31 		&\ 3.36 		&\ 5.61  			&\ 7.94  		&\red{\ 10.00}	&\red{\ 11.34}	&\red{\ 11.55} \\
 
 		& {3M}  &\ 2.50			&\ 5.92		&\red{\ 9.21}		&\red{\ 11.35}	&\red{\ 11.31}	&\red{\ 8.24}	&\red{\ 1.66} \\
 		
 		& {6M}  &\ 4.09			&\ 7.93		&\red{\ 11.34}		&\red{\ 13.02}	&\red{\ 11.87}	&\red{\ 7.04}	&\red{-2.11} \\
 		
		& {1Y}  &\ 5.10			&\ 8.18		&\red{\ 11.05}		&\red{\ 12.73}	&\red{\ 12.43}	&\red{\ 9.50}	&\red{\ 3.39} \\
\midrule

 Put 10 	& {1M}  &\ 0.47			&\ 4.78  		&\ 10.97 			&\ 19.32 		&\red{\ 29.81} 	&\red{\ 42.31} 	&\red{\ 56.54} \\

		& {3M}  &\ 2.10			&\ 10.46		&\red{\ 22.57}		&\red{\ 38.09}	&\red{\ 56.38}	&\red{\ 76.63}	&\red{\ 98.37} \\
		
		& {6M}  &\ 4.70			&\ 14.43		&\red{\ 28.08}		&\red{\ 44.95}	&\red{\ 64.13}	&\red{\ 84.85}	&\red{\ 106.37} \\
		
 		& {1Y}  &\ 5.13			&\ 12.44		&\red{\ 22.53}		&\red{\ 34.81}	&\red{\ 48.54}	&\red{\ 63.19}	&\red{\ 78.41} \\
 \bottomrule
\end{tabular}
\end{center}
\end{table}

%%--------Varying rho in Heston model
\label{sec:varyingrhoheston}
\wo

%---Table: Signed errors in the Heston model with rho varying (Done and checked!)
\begin{table}[H]
\caption{Signed error of implied volatilities in basis points, computed in the Heston model with $\rho$ varying from 160\% to 40\% of its ``safe'' 3-piece value. }
\label{table:varyingrhoheston}
\begin{center}
\begin{tabular}{ll lllllll} 
\toprule
	&	{$\rho$} & \ 160\% & \ 140\%  & \ 120\%  & \ 100\% & \ 80\% & \ 60\% & \ 40\% \\
\midrule

ATM		& {1M}  &\ 29.14		&\ 14.93 		&\ 6.37 		&\ 1.16		&-2.02 \wo	&-3.96 \wo	&-5.10 \wo \\

 		& {3M}  &\red{\ 27.01}	&\red{\ 10.49}	&\red{-0.06}	&\red{-6.97}	&\red{-11.56}	&\red{-14.60}	&\red{-16.52} \\
 		
 		& {6M}  &\red{\ 24.82}	&\red{\ 8.30}	&\red{-2.34}	&\red{-9.41}	&\red{-14.19}	&\red{-17.41}	&\red{-19.49} \\
 		
 		& {1Y}  &\red{\ 25.35}	&\red{\ 10.58}	&\red{\ 1.42}	&\red{-4.47}	&\red{-8.34}	&\red{-10.91}	&\red{-12.54} \\
\midrule

 Put 25	& {1M}  &\ 44.10	    	&\ 28.02		&\ 16.46		&\ 8.28 		&\ 2.52 		&-1.49		&-4.19 \\
 
	 	& {3M}  &\red{\ 63.28}	&\red{\ 40.72}	&\red{\ 23.85}	&\red{\ 11.21}	&\red{\ 1.67}	&\red{-5.51}	&\red{-10.85} \\
	 	
 		& {6M}  &\red{\ 68.04}	&\red{\ 44.03}	&\red{\ 26.03}	&\red{\ 12.39}	&\red{\ 1.96}	&\red{-6.03}	&\red{-12.12} \\
 		
		& {1Y}  &\red{\ 60.63}	&\red{\ 39.51}	&\red{\ 24.27}	&\red{\ 13.08}	&\red{\ 4.73}	&\red{-1.55}	&\red{-6.27} \\
\midrule

 Put 10 	& {1M}  &\ 28.21		&\ 28.60		&\ 24.52		&\ 18.96 		&\ 13.15		&\ 7.68		&\ 2.81 \\
 
 		& {3M}  &\red{\ 64.31}	&\red{\ 58.82}	&\red{\ 49.05}	&\red{\ 38.07}	&\red{\ 27.19}	&\red{\ 17.01}	&\red{\ 7.78} \\
 		
 		& {6M}  &\red{\ 76.48}	&\red{\ 68.62}	&\red{\ 56.94}	&\red{\ 44.37}	&\red{\ 32.11}	&\red{\ 20.69}	&\red{\ 10.33} \\
 		
 		& {1Y}  &\red{\ 64.00}      &\red{\ 55.57}	&\red{\ 44.91}	&\red{\ 34.16}	&\red{\ 24.11}	&\red{\ 15.04}	&\red{\ 7.05} \\
 \bottomrule
\end{tabular}
\end{center}
\end{table}

The sensitivity analysis is consistent with what we expect. For example, for large maturity $T$, it can be easily shown that the component-wise variance of the difference in the expansion and evaluation point increases. Moreover, for large vol-of-vol $\lambda$ or large correlation $|\rho|$, one would expect the variance of the difference in the expansion and evaluation point to increase. Thus, when these parameters are large, we expect our closed-form approximation formula to be less accurate. Indeed, the numerical results confirm this behaviour. However, there seems to be a peculiarity when comparing the numerical results pertaining to the maturities $T = 1/2$ and $T = 1$. Specifically, it seems that the approximation formula for $T = 1$ is often more accurate than for $T = 1/2$. By studying the error bound from \Cref{thm:errbound}, one explanation for this behaviour is that the terms that are exponentially decaying in $T$ from the function $\zeta(T)$ (recall this is the function $M(T,K)$ for fixed $K$, see \Cref{parblowu}) may start to contribute more around this region of $T$. Of course, this is just conjecture, as such a claim is dependent on all the other terms in the error bound \Cref{thm:errbound} (most notably, the behaviour of the moments of the underlying variance process). However, it seems like a reasonable claim, as the other terms in the error bound are all increasing functions of $T$. A precise quantification of this claim would require further investigation of the moments of the underlying variance process. In terms of application, for realistic parameter values (that is, values around the ``safe'' 3-piece parameter sets) we see that the magnitude of error is around 10-50bps, which would be deemed reasonable for use in industry.

%-------------------------------------------------------------------------GARCH sensitivity analysis---------------------------------------------------------------------------
\subsection{GARCH diffusion model sensitivity analysis}
We start from the same ``safe'' 3-piece parameter set from \Cref{hestonsense}, albeit with $\rho = 0$ on each piece. In our analysis, we vary one of the 3-piece parameters $(\kappa, \theta, \lambda)$ at a time and keep the rest fixed. We employ the same strategy as in the Heston framework, we start at 40\% of the safe 3-piece parameter value, and increase it in increments of 20\%, all the way up to 160\%, whilst the rest are fixed at their safe value. The relevant tables are \Cref{table:varyingkappaGARCH}, \Cref{table:varyingthetaGARCH}, and \Cref{table:varyinglambdaGARCH} for the analysis of $\kappa$, $\theta$, and $\rho$ respectively.

%------------------------------------------------------------------------------GARCH tables-----------------------------------------------------------------------------------

%%---------Varying kappa in GARCH model 
%\label{sec:varyingkappaGARCH}
%\wo

%---Table: Signed errors in the GARCH diffusion model with kappa varying (Done and checked!)
\begin{table}[H]
\caption{Signed error of implied volatilities in basis points, computed in the GARCH diffusion model with $\kappa$ varying from 40\% to 160\% of its ``safe'' 3-piece value.}
\label{table:varyingkappaGARCH}
\begin{center}
\begin{tabular}{ll lllllll} 
\toprule
	&	{$\kappa$} & \ 40\% & \ 60\%  & \ 80\%  & \ 100\% & \ 120\% & \ 140\% & \ 160\% \\
\midrule

 ATM		& {1M} 	&-1.33 \wo	&-1.70 \wo	&-2.06 \wo	&-2.35	&-2.58	&-2.77	&-2.93 \\
 		& {3M}  	&-2.69		&-3.10		&-3.30		&-3.31	&-3.22	&-3.06	&-2.87 \\
 		& {6M}  	&-3.32		&-3.16		&-2.74		&-2.25	&-1.78	&-1.36	&-1.02 \\
 		& {1Y}  	&-2.79		&-1.80		&-1.04		&-0.54	&-0.25	&-0.08	&\ 0.00 \\
 \midrule

 Put 25 	& {1M}  	&-1.36		&-1.73		&-2.09		&-2.38	& -2.61	&-2.81	&-2.96 \\
 		& {3M}  	&-2.71		&-3.12		&-3.31		&-3.32	& -3.23	&-3.07	&-2.87 \\
		& {6M}  	&-3.30		&-3.13		&-2.72		&-2.23	&-1.75	&-1.34	&-1.00 \\
		& {1Y}   	&-2.77		&-1.77		&-1.02		&-0.52	&-0.23	&-0.07	&\ 0.02 \\
 \midrule

 Put 10	& {1M}  	&-1.37		&-1.74		&-2.10		&-2.39	&-2.62	&-2.81	&-2.97 \\
 		& {3M}  	&-2.75		&-3.16		&-3.35		&-3.37	&-3.27	&-3.11	&-2.91 \\
 		& {6M}  	&-3.33		&-3.17		&-2.75		&-2.25	&-1.78	&-1.36	&-1.01 \\
 		& {1Y}   	&-2.76		&-1.76		&-1.00		&-0.51	&-0.22	&-0.06	&\ 0.02 \\
 \bottomrule
\end{tabular}
\end{center}
\end{table}

%%--------Varying theta in GARCH model 
\label{sec:varyingthetaGARCH}
\wo

%---Table: Signed errors in the GARCH diffusion model with theta varying (Done and checked!)
\begin{table}[H]
\caption{Signed error of implied volatilities in basis points, computed in the GARCH diffusion model with $\theta$ varying from 40\% to 160\% of its ``safe'' 3-piece value.}
\label{table:varyingthetaGARCH}
\begin{center}
\begin{tabular}{ll lllllll} 
\toprule
	&	{$\theta$} & \ 40\% & \ 60\%  & \ 80\%  & \ 100\% & \ 120\% & \ 140\% & \ 160\% \\
\midrule

 ATM		& {1M}  	&-1.09 \wo	&-1.51 \wo	&-1.95 \wo	&-2.35	&-2.71	&-3.04	&-3.35 \\
 		& {3M}  	&-1.43		&-2.14		&-2.78		&-3.34	&-3.85	&-4.31	&-4.73 \\
 		& {6M}  	&-0.67		&-1.29		&-1.82		&-2.27	&-2.67	&-3.02	&-3.35 \\
 		& {1Y}       &\ 0.47		&\ 0.08		&-0.22		&-0.46	&-0.66	&-0.84	&-1.00 \\
 \midrule

 Put 25 	& {1M}  	&-1.11		&-1.54		&-1.98		&-2.38	&-2.74	&-3.07	&-3.38 \\
 		& {3M}  	&-1.42		&-2.13		&-2.77		&-3.34	&-3.84	&-4.30	&-4.73 \\
 		& {6M}  	&-0.65		&-1.26		&-1.79		&-2.24	&-2.63	&-2.99	&-3.31 \\
		& {1Y}      &\ 0.43		&\ 0.04		&-0.26		&-0.51	&-0.72	&-0.90	&-1.07 \\
\midrule

Put 10 	& {1M} 	&-1.13		&-1.56		&-2.00		&-2.39	&-2.75	&-3.09	&-3.40 \\
		& {3M}  	&-1.43		&-2.14		&-2.78		&-3.35	&-3.85	&-4.31	&-4.74 \\
 		& {6M}  	&-0.69		&-1.30		&-1.82		&-2.27	&-2.66	&-3.02	&-3.34 \\
 		& {1Y}  	&\ 0.40		&\ 0.02		&-0.29		&-0.53	&-0.75	&-0.93	&-1.10 \\
 \bottomrule
\end{tabular}
\end{center}
\end{table}

%%--------Varying lambda in GARCH model 
\label{sec:varyinglambdaGARCH}
\wo

%---Table: Signed errors in the GARCH diffusion model with lambda varying (Done and checked!)
\begin{table}[H]
\caption{Signed error of implied volatilities in basis points, computed in the GARCH diffusion model with $\lambda$ varying from 40\% to 160\% of its ``safe'' 3-piece value.}
\label{table:varyinglambdaGARCH}
\begin{center}
\begin{tabular}{ll lllllll} 
\toprule
	&	{$\lambda$} & \ 40\% & \ 60\%  & \ 80\%  & \ 100\% & \ 120\% & \ 140\% & \ 160\% \\
\midrule

 ATM 	& {1M}  	&-2.51 \wo	&-2.36 \wo	&-2.36 \wo	&-2.36	&-2.36	&-2.36	&-2.35 \\
 		& {3M}  	&-3.43		&-3.33		&-3.32		&-3.32	&-3.31	&-3.31	&-3.31 \\
 		& {6M}  	&-2.33		&-2.26		&-2.26		&-2.26	&-2.26	&-2.26	&-2.27 \\
 		& {1Y}  	&-0.55		&-0.51		&-0.50		&-0.50	&-0.50	&-0.50	&-0.51 \\
 \midrule

 Put 25 	& {1M}  	&-2.51		&-2.36		&-2.36		&-2.35	&-2.35	&-2.35	&-2.34 \\
 		& {3M}  	&-3.42		&-3.31		&-3.31		&-3.30	&-3.29	&-3.29	&-3.28 \\
 		& {6M}  	&-2.32		&-2.26		&-2.25		&-2.25	&-2.25	&-2.25	&-2.25 \\
 		& {1Y}  	&-0.57		&-0.54		&-0.55		&-0.56	&-0.57	&-0.59	&-0.61 \\
\midrule

 Put 10 	& {1M}  	&-2.51		&-2.36		&-2.36		&-2.36	&-2.36	&-2.35	&-2.35 \\
 		& {3M}  	&-3.43		&-3.33		&-3.32		&-3.32	&-3.32	&-3.31	&-3.31 \\
 		& {6M}  	&-2.33		&-2.26		&-2.26		&-2.26	&-2.26	&-2.26	&-2.25 \\
 		& {1Y}  	&-0.55		&-0.50		&-0.50		&-0.49	&-0.49	&-0.48	&-0.48 \\
 \bottomrule
\end{tabular}
\end{center}
\end{table}

The error in implied volatility for the closed-form approximation formula in the GARCH diffusion model  behaves well, with most being less than 3bp in magnitude. In contrast to the Heston model analysis, this is most likely due to the fact that the correlation $\rho$ is assumed to be 0 always, and hence one would expect the variance of the difference in the expansion and evaluation point to be smaller. Otherwise, the approximation behaves as we expect, with errors being larger for large maturity $T$ and vol-of-vol $\lambda$, as the variance of the difference in the expansion and evaluation point will grow with these parameters. Again, there is an anomaly when comparing the numerical results pertaining to the maturities $T = 1/2$ and $T = 1$, as the latter seems to be more accurate. For similar reasons to the Heston model case, we attribute this behaviour to the function $\zeta(T)$ (see \Cref{parblowu} and \Cref{cor:errboundrho0}).

%------------------------------------------------------------------------------Computational time-----------------------------------------------------------------------------
\subsection{Computational runtime analysis}
The use of our closed-form approximation formula yields approximations of put option prices essentially instantaneously. However, it is worth investigating the speed up/slow down when considering $n$-piece parameters for different $n$. We will consider pricing a put option with maturities 6M and 1Y, whilst varying the number of pieces the parameters possess in both the Heston and GARCH diffusion models va our closed-form approximation formulas. The values of the parameters do not affect the runtime, and thus we have not provided them in our numerical tests. Each runtime we display will be the average of 100 runtimes with identical parameters.

%---------------------------Table: Runtimes Heston
\begin{table}[H]
\caption{Runtimes for pricing a put option with maturities 6M, 1Y in the Heston model via our closed-form approximation method, where \#pieces of the parameters vary. The length of each time-interval corresponding to a piece is equal. Each runtime is the average of 100 runtimes (with identical parameters).}
\label{table:Hestonruntime}

\begin{center}
\begin{tabular}{cc c ll}
\toprule

& 	 \#Pieces && 		Run(ms) &  Slow down($\times$) \\
\midrule

6M 	& 1	&&	 	0.77   & 1.00  \\
 	& 2	&&	 	1.43   & 1.87  \\
 	& 3	&&	 	2.48   & 3.24  \\
 	& 4	&&	 	4.70   & 6.13  \\
	& 5	&&	 	7.87   & 10.27  \\
 	& 6	&&	 	11.05   & 14.43  \\
 	& 7	&&	 	16.98   & 22.18  \\
 	& 8	&&	 	22.21   & 29.00  \\
 	& 9	&&	 	31.53   & 41.17  \\
 	& 10	&&	 	40.73   & 53.19  \\
\midrule

1Y 	& 1	&&	 	0.80   & 1.00  \\
 	& 2	&&	 	1.71   & 2.13  \\
 	& 3	&&	 	2.65   & 3.30  \\
 	& 4	&&	 	4.60   & 5.74  \\
	& 5	&&	 	6.59   & 8.21  \\
 	& 6	&&	 	10.94   & 13.64  \\
 	& 7	&&	 	15.80   & 19.69  \\
 	& 8	&&	 	22.17   & 27.65  \\
 	& 9	&&	 	29.75   & 37.09  \\
 	& 10	&&	 	40.08   & 49.97  \\
\bottomrule
\end{tabular}
\end{center}
\end{table}

%---------------------------Table: Runtimes GARCH
\begin{table}[H]
\caption{Runtimes for pricing a put option with maturities 6M, 1Y in the GARCH diffusion model via our closed-form approximation method, where the \#pieces of parameters vary. The length of each time-interval corresponding to a piece is equal. Each runtime is the average of 100 runtimes (with identical parameters).}
\label{table:GARCHruntime}

\begin{center}
\begin{tabular}{cc c ll}
\toprule

& 	 \#Pieces && 		Run(ms) &  Slow down($\times$) \\
\midrule

6M 	& 1	&&	 	0.69   & 1.00  \\
 	& 2	&&	 	1.37   & 1.98  \\
 	& 3	&&	 	2.81   & 4.05  \\
 	& 4	&&	 	4.56   & 6.58  \\
	& 5	&&	 	8.30   & 11.96  \\
 	& 6	&&	 	13.07   & 18.84  \\
 	& 7	&&	 	20.28   & 29.22  \\
 	& 8	&&	 	30.56   & 44.03  \\
 	& 9	&&	 	44.38   & 63.94  \\
 	& 10	&&	 	62.86   & 90.57  \\
\midrule

1Y 	& 1	&&	 	0.67   & 1.00  \\
 	& 2	&&	 	1.45   & 2.17  \\
 	& 3	&&	 	2.63   & 3.93  \\
 	& 4	&&	 	5.08   & 7.59  \\
	& 5	&&	 	8.75   & 13.07  \\
 	& 6	&&	 	12.85   & 19.19  \\
 	& 7	&&	 	20.79   & 31.04  \\
 	& 8	&&	 	29.50   & 44.05  \\
 	& 9	&&	 	44.24   & 66.07  \\
 	& 10	&&	 	65.43   & 97.71  \\
\bottomrule
\end{tabular}
\end{center}
\end{table}

The Heston model runtime analysis is given in \Cref{table:Hestonruntime}, whilst the GARCH diffusion model runtime analysis is given in \Cref{table:GARCHruntime}. Unsurprisingly, the maturity value does not really affect the runtime, as varying the maturity does not contribute to any computational burden in the formula. For 1-piece parameters (that is, constant parameters), the runtimes of the approximation formulas for both the Heston and GARCH diffusion models are almost the same (just less than a millisecond). However, for parameters with a large number of pieces, the runtime of the GARCH diffusion model formula is longer than for the Heston model. For example, for 10-piece parameters, running the Heston model closed-form approximation formula takes around 40 milliseconds, whereas the GARCH diffusion model one takes around 65 milliseconds. This may seem surprising, as one may expect that the GARCH diffusion model closed-form approximation formula would be less costly to run, since here we are in the simpler setting where $\rho = 0$ a.e.. However, the reason why the GARCH diffusion model's closed-form approximation formula is computationally more costly is because it possesses a 5-fold integral operator term, whereas the Heston model closed-form approximation formula has only 4-fold and below integral operators.

%-----------------------------------------------------------------------------------Conclusion------------------------------------------------------------------------------
%-----------------------------------------------------------------------------------------------------------------------------------------------------------------------------
\section{Conclusion}
\noindent
\label{sec:conclusion}
We have derived closed-form approximation formulas for European put option prices in the context of stochastic volatility models with time-dependent parameters. Our method involves a second-order Taylor expansion of the mixing solution, followed by simplification of a number of expectations via the use of change of measure techniques. Such a method has been considered by Drimus with respect to the Heston model with constant parameters \citep{drimus2011closed}. We extend this method as we consider time-dependent parameters, which requires an additional approximation of one of the expectations. Furthermore, we obtain the expression for the error term induced by the expansion. We give general bounds on the induced error term in terms of higher order moments of the underlying variance process. Under the Heston framework, we find that all terms induced by the expansion are able to be expressed as iterated integrals. Furthermore, we attempt to generalise this method to the GARCH diffusion model. We show that obtaining an expression for the solution or moments of the resulting variance process after the change of measure is non-trivial to achieve. By assuming $\rho = 0$ a.e. in the GARCH diffusion model, we are able to address this problem, albeit with the added assumption of uncorrelated spot and volatility movements. Additionally, we show that the iterated integrals obey a convenient recursive property when parameters are piecewise-constant. By doing so, the approximation formulas become closed-form. In addition, we devise a fast calibration scheme which exploits the recursive property of our integral operators. Lastly, we perform a numerical error and sensitivity analysis to investigate the quality of our approximation in the Heston and GARCH diffusion models. We find that the error is well within the range appropriate for application purposes and behaves as we expect for certain parameter values, such as long maturity, large vol-of-vol and large correlation. 

The purely probabilistic mixing solution approach, which is the backbone of our expansion method, is very appealing due to its generality and ability to handle time-dependent parameters. Further research would be needed to combine it, with no correlation restriction, with the type of non-affine stochastic volatility models favoured by practitioners.

%--------------------------References
\renewcommand{\bibname}{References}
\bibliographystyle{plainurl}
%\bibliography{../References}
\bibliography{References}
\appendix

%------------------------------------------------------------------------------Appendix A Mixing solution-------------------------------------------------------------------
%----------------------------------------------------------------------------------------------------------------------------------------------------------------------------------------
\section{Mixing solution}
\noindent
\label{appen:mixing}
In this appendix, we give a derivation of the result referred to as the mixing solution by \citep{hull1987pricing}. This result is crucial for the expansion methodology implemented in \Cref{sec:prelims2}. Hull and White first established the expression for the case of independent Brownian motions $W$ and $B$. Later on, this was extended for the correlated Brownian motions case, see \citep{Willard97,romano1997contingent}. 

Under a domestic risk-neutral measure $\Qro$, suppose that the spot $S$ with variance $\sigma$ follows the dynamics
\begin{align*}
	\dd S_t &= S_t ( (r_t^d - r_t^f ) \dd t + \sqrt{\sigma_t }\dd W_t ), \quad S_0, \\
	\dd \sigma_t &= \alpha(t, \sigma_t) \dd t + \beta(t, \sigma_t) \dd B_t, \quad \sigma_0,  \\
	\dd \langle W, B \rangle_t &= \rho_t \dd t.
\end{align*}

%-------Thm: mixing solution
\begin{theorem}[Mixing solution]
\label{thm:mixing}
Let $\text{Put} = e^{-\int_0^T r_t^d \dd t } \Ex  ( K- S_T)_+$. Then
\begin{align*}
	\text{Put} &= \Ex \left (  e^{-\int_0^T r_t^d \dd t } \Ex \big  [ (K - S_T)_+ | \FF_T^B \big ] \right ) \\
	&= \Ex \left ( \text{Put}_{\text{BS}} \left ( S_0 \xi_T, \int_0^T \sigma_t  ( 1 - \rho_t^2) \dd t \right ) \right ),
\end{align*}
where $\text{Put}_{\text{BS}}$ is given in \cref{PBS}.
\end{theorem}
%-------Proof of mixing solution
\begin{proof}
By writing the driving Brownian motion of the spot as $W_t = \int_0^t \rho_u  \dd B_u +  \int_0^t \sqrt{1 - \rho_u^2} \dd Z_u$, where $Z$ is a Brownian motion under $\Qro$ which is independent of $B$, this gives the explicit pathwise unique strong solution of $S$ as 
\begin{align*}
	S_T &= S_0 \xi_T \exp \left \{   \int_0^T (r_t^d - r_t^f ) \dd t - \frac{1}{2} \int_0^T \sigma_t (1 - \rho_t^2) \dd t + \int_0^T  \sqrt{ \sigma_t  (1 - \rho_t^2)} \dd Z_t \right \}, \\
	\xi_t &:= \exp \left \{ \int_0^t \rho_u \sqrt{ \sigma_u}  \dd B_u - \frac{1}{2} \int_0^t \rho_u^2  \sigma_u \dd u \right \}. 
\end{align*}
First, notice that both $ \sigma $ and $\xi$ are adapted to the filtration $(\FF_t^B)_{0 \leq t \leq T}$. Thus, it is evident that $S_T|\FF_T^B$ will have a log-normal distribution, namely
\begin{align*}
	S_T | \FF_T^B &\sim \mathcal{LN}\left ( \tilde \mu (T) , \tilde \sigma^2(T) \right ), \\
	\tilde \mu(T) &:=  \ln(S_0 \xi_T) +  \int_0^T (r_t^d - r_t^f ) \dd t - \frac{1}{2} \int_0^T  \sigma_t  (1 - \rho_t^2) \dd t,\\
	\tilde \sigma^2(T) &:=  \int_0^T  \sigma_t  (1 - \rho_t^2) \dd t.
\end{align*}
Hence, the calculation of $e^{-\int_0^T r_t^d \dd t } \Ex \big [(K - S_T)_+ | \FF_T^B \big ]$ will result in a Black-Scholes like formula. 
\begin{align*}
&e^{-\int_0^T r_t^d \dd t } \Ex \big [(K - S_T)_+ | \FF_T^B \big ] \\ &= K e^{ -\int_0^T r_t^d \dd t } \NN \left  (\frac{ \ln(K) - \tilde \mu(T) }{\tilde \sigma(T)}\right ) - e^{-\int_0^T r_t^d \dd t } e^{\tilde \mu(T) + \frac{1}{2} \tilde \sigma^2(T) } \NN \left  ( \frac{ \ln(K) - \tilde \mu(T)-  \tilde \sigma^2(T)}{\tilde \sigma(T)}\right ) \\
&= K e^{ -\int_0^T r_t^d \dd t }  \NN \left  ( \frac{ \ln(K) - \tilde \mu(T) - \frac{1}{2}   \tilde \sigma^2(T)}{\tilde \sigma(T)} + \frac{1}{2} \tilde \sigma(T)\right ) \\&- S_0 \xi_T e^{-\int_0^T r_t^f \dd t }\NN \left  (\frac{\ln(K) - \tilde \mu(T)  - \frac{1}{2} \tilde \sigma^2(T)}{\tilde \sigma(T)} - \frac{1}{2} \tilde \sigma(T)\right ) \\
&= K e^{ -\int_0^T r_t^d \dd t } \NN \left  ( \frac{\ln(K/S_0 \xi_T )  -\int_0^T (r_t^d - r_t^f) \dd t}{\tilde \sigma(T)} + \frac{1}{2} \tilde \sigma(T)\right ) \\&- S_0 \xi_T e^{-\int_0^T r_t^f \dd t }\NN \left  ( \frac{\ln(K/S_0 \xi_T )  -\int_0^T (r_t^d - r_t^f) \dd t}{\tilde \sigma(T)} - \frac{1}{2} \tilde \sigma(T)\right ).
\end{align*}
It is immediate that $ e^{-\int_0^T r_t^d \dd t } \Ex \big [(K- S_T)_+ | \FF_T^B \big ] = \text{Put}_{\text{BS}} \left ( S_0 \xi_T, \tilde \sigma^2(T) \right )$.
\end{proof}

%------------------------------------------------------------------------------Appendix B PutBS------------------------------------------------------------------------------
%--------------------------------------------------------------------------------------------------------------------------------------------------------------------------------
\section{$\text{Put}_{\text{BS}}$ partial derivatives}
\noindent
\label{appen:greeks}
This appendix contains some partial derivatives for the Black-Scholes put option formula $\text{Put}_{\text{BS}}$, given in \cref{PBS}. One can think of these partial derivatives as being analogous to Black-Scholes Greeks. However, these are slightly different as our Black-Scholes formulas are parameterised with respect to integrated variance rather than volatility.

%-----------------------Second order PutBS
\subsection{First-order $\textnormal{Put}_{\textnormal{BS}}$}
\begin{align*}
\partial_{x} \text{Put}_{\text{BS}} &= e^{- \int_0^T r_u^f \dd u } \left (\NN(d_+) - 1  \right ), \\
\partial_{y} \text{Put}_{\text{BS}} &= \frac{x e^{- \int_0^T r_u^f \dd u } \phi(d_+) }{2 \sqrt{y}}.
\end{align*}

%-----------------------Second order PutBS
\subsection{Second-order $\textnormal{Put}_{\textnormal{BS}}$}
\begin{align*}
\partial_{xx} \text{Put}_{\text{BS}} &= \frac{e^{- \int_0^T r_u^f \dd u } \phi(d_+) }{x \sqrt{y}}, \\
\partial_{yy} \text{Put}_{\text{BS}} &= \frac{xe^{-\int_0^T r_u^f \dd u } \phi( d_+)}{4 y^{3/2}} (d_-d_+ - 1),  \\
\partial_{xy} \text{Put}_{\text{BS}} &= (-1)\frac{ e^{ - \int_0^T  r_u^f \dd u } \phi( d_+) d_-}{2 y}.
\end{align*}

%-----------------------Third order PutBS
\subsection{Third-order $\textnormal{Put}_{\textnormal{BS}}$}
\begin{align*}
\partial_{xxx} \text{Put}_{\text{BS}} &= (-1) \frac{e^{- \int_0^T r_u^f \dd u } \phi(d_+) }{x^2 y} (d_+ + \sqrt{y}), \\
\partial_{xxy} \text{Put}_{\text{BS}} &=   \frac{e^{- \int_0^T r_u^f \dd u } \phi(d_+) }{2 y} (d_- d_+ - 1), \\
\partial_{xyy} \text{Put}_{\text{BS}} &= (-1)\frac{ e^{ - \int_0^T  r_u^f \dd u } \phi( d_+) }{2 y^2} \left (  \frac{d_-^2 d_+}{2} - \frac{d_+ }{2}- d_- \right ), \\
\partial_{yyy} \text{Put}_{\text{BS}} &= \frac{ x e^{ - \int_0^T  r_u^f \dd u } \phi( d_+)}{8 y^{5/2}} \left ( (d_- d_+ - 1)^2 - (d_- + d_+)^2 + 2  \right ).
\end{align*}

%-----------------------Fourth order PutBS
\subsection{Fourth-order $\textnormal{Put}_{\textnormal{BS}}$} 
\begin{align*}
\partial_{xxxx} \text{Put}_{\text{BS}} &= \frac{e^{- \int_0^T r_u^f \dd u } \phi(d_+) }{x^3 y^{3/2}} (d_+^2 + 3d_+ \sqrt{y} + 2y +1), \\
\partial_{xxxy} \text{Put}_{\text{BS}} &=   \frac{e^{- \int_0^T r_u^f \dd u } \phi(d_+) }{2 x y^{3/2}} (d_- (1- d_+^2)), \\
\partial_{xxyy} \text{Put}_{\text{BS}} &= (-1)\frac{ e^{ - \int_0^T  r_u^f \dd u } \phi( d_+) }{2 x y^{5/2}} \left ( \frac{1}{2} (d_- + d_+)^2 + d_-d_+ \left(1 - \frac{d_- d_+}{2}\right)  - \frac{3}{2}\right ),\\
\partial_{xyyy} \text{Put}_{\text{BS}} &= \frac{ e^{ - \int_0^T  r_u^f \dd u } \phi( d_+)}{8 y^{3}} \Big ( (\sqrt{y} - d_+) \left [(d_- d_+ - 1)^2 - (d_- + d_+)^2 + 2 \right ] \\&+ 4 \left [d_+ (d_- -  d_+) - d_- - d_+ \right ] \Big ), \\
\partial_{yyyy} \text{Put}_{\text{BS}} &= \frac{ x e^{ - \int_0^T  r_u^f \dd u } \phi( d_+)}{8 y^{7/2}} \Bigg ( \frac{1}{2} (d_- d_+ - 1)^2(d_- d_+ - 5) - (d_-d_+ - 1)(d_- + d_+) \\&- \frac{1}{2}(d_- + d_+)^2 (d_-d_+ - 7) + (d_-d_+ - 1)   \Bigg ).
\end{align*}

%------------------------------------------------------------------------------Appendix C Greeks----------------------------------------------------------------------------
%--------------------------------------------------------------------------------------------------------------------------------------------------------------------------------
\section{Greeks}
\noindent
\label{appen:Greeks}
Expressions for second-order approximations of put option Greeks can be obtained via simple partial differentiation of $\text{Put}^{(2)}$ (\cref{price2}).

%-------Remark: Greeks
The Put Delta approximation is obtained via partial differentiation of $\text{Put}^{(2)}$ with respect to the underlying $S_0$. 
\begin{align*}
\partial_{S_0} \text{Put}^{(2)} &=  \partial_{x} \text{Put}_{\text{BS}} ( \hat x, \hat y) \\ &+ \frac{1}{2}  \left [ 2 S_0 \partial_{xx} \text{Put}_{\text{BS}} ( \hat x, \hat y) + S_0^2 \partial_{xxx} \text{Put}_{\text{BS}} ( \hat x, \hat y)  \right ]\Ex (\xi _T - 1)^2 \\ &+ \frac{1}{2}\partial_{xyy} \text{Put}_{\text{BS}} ( \hat x, \hat y)   \Ex \left (\int_0^T (1 - \rho_t^2) (\sigma_t - \Ex(\sigma_t)) \dd t \right )^2 \\ 
&+ \left [\partial_{xy} \text{Put}_{\text{BS}} ( \hat x, \hat y) + S_0  \partial_{xxy} \text{Put}_{\text{BS}} ( \hat x, \hat y)  \right ] \Ex \left \{  ( \xi_T - 1) \left (\int_0^T (1 - \rho_t^2) (\sigma_t - \Ex(\sigma_t)) \dd t \right ) \right \}.
\end{align*}
The Put Gamma approximation is obtained via partial differentiation of the Put Delta approximation with respect to the underlying $S_0$. 
\begin{align*}
\partial_{S_0 S_0} \text{Put}^{(2)} &=  \partial_{xx} \text{Put}_{\text{BS}} ( \hat x, \hat y) \\ &+ \frac{1}{2}  \left [ 2 \partial_{xx} \text{Put}_{\text{BS}} ( \hat x, \hat y) +  2 S_0 \partial_{xxx} \text{Put}_{\text{BS}} ( \hat x, \hat y)  + S_0^2 \partial_{xxxx} \text{Put}_{\text{BS}} ( \hat x, \hat y) \right ]\Ex (\xi _T - 1)^2 \\ &+ \frac{1}{2}\partial_{xxyy} \text{Put}_{\text{BS}} ( \hat x, \hat y) \Ex \left (\int_0^T (1 - \rho_t^2) (\sigma_t - \Ex(\sigma_t)) \dd t \right )^2 \\ 
&+ \left [2\partial_{xxy} \text{Put}_{\text{BS}} ( \hat x, \hat y) + S_0  \partial_{xxxy} \text{Put}_{\text{BS}} ( \hat x, \hat y)  \right ] \\&\cdot \Ex \left \{  ( \xi_T - 1) \left (\int_0^T (1 - \rho_t^2) (\sigma_t - \Ex(\sigma_t)) \dd t \right ) \right \}.
\end{align*}
Notice that the above expectations are the same as those from \crefrange{mom1}{mom3}. These expectations are made explicit in \Cref{sec:mom}.

%-------------------------------------------------------------------------------Appendix D calculating moments-----------------------------------------------------------
%-------------------------------------------------------------------------------------------------------------------------------------------------------------------------------
\section{Calculation of moments}
\noindent
\label{appen:moments}
In this appendix, we derive expressions for some of the moments, mixed moments and covariances of the CIR and IGa processes utilised in this article. The results for the CIR process are well known, however, we have not found a source for the IGa moments with time-dependent parameters in the literature.
%--------Deriving the CIR moments
\subsection{Moments of the CIR process}
\label{CIRappend}
Let $V$ be a CIR$( v_0; \kappa_t, \theta_t, \lambda_t)$. It satisfies the SDE 
\begin{align*}
\dd V_t = \kappa_t(\theta_t -  V_t) \dd t + \lambda_t  \sqrt{V_t} \dd B_t, \quad V_0 = v_0,
\end{align*}
where we assume $(\kappa_t)_{0 \leq t \leq T}, (\theta_t)_{0 \leq t \leq T}$ and $(\lambda_t)_{0 \leq t \leq T}$ are time-dependent, deterministic, strictly positive and bounded on $[0,T]$. 
For $s<t$, it can be integrated to obtain 
\begin{align}
 V_t =  V_s e^{- \int_s^t \kappa_z \dd z } + \int_s^t e^{- \int_u^t \kappa_z \dd z} \kappa_u \theta_u \dd u + \int_s^t e^{- \int_u^t \kappa_z \dd z } \lambda_u \sqrt{ V_u} \dd B_u. \label{CIRsol}
\end{align}
In particular, for $s=0$, 
\begin{align}
 V_t =  v_0 e^{- \int_0^t \kappa_z \dd z } + \int_0^t e^{-\int_u^t \kappa_z \dd z } \kappa_u \theta_u \dd u + \int_0^t e^{-\int_u^t \kappa_z \dd z } \lambda_u \sqrt{ V_u} \dd B_u. \label{CIRsol0}
\end{align}

%-------Prop: CIR moments
\begin{proposition}
$V$ has the following moments: 
\begin{align*}
\Ex(V_t^n) &= e^{- \int_0^t n\kappa_z \dd z }  \left ( v_0^n + \int_0^t e^{ \int_0^u n \kappa_z \dd z }   \left (n \kappa_u \theta_u + \frac{1}{2} n(n-1) \lambda_u^2  \right ) \Ex(V_u^{n-1} ) \dd u \right ),\\
\Var( V_t) &= \int_0^t \lambda_u^2 e^{-2 \int_u^t \kappa_z \dd z} \left \{ v_0 e^{- \int_0^u \kappa_z \dd z } + \int_0^u e^{- \int_p^u \kappa_z \dd z }\kappa_p \theta_p \dd p \right \} \dd u, \\
\Cov( V_s,  V_t) &= e^{-\int_s^t \kappa_z \dd z } \int_0^s \lambda_u^2 e^{-2 \int_u^s \kappa_z \dd z} \left \{ v_0 e^{- \int_0^u \kappa_z \dd z } + \int_0^u e^{- \int_p^u \kappa_z \dd z }\kappa_p \theta_p \dd p \right \} \dd u, \\
\Ex(V_s^m V_t^n) &= e^{-\int_0^t n \kappa_z \dd z } \left ( \Ex(V_s^{m+n}) + \int_s^t e^{ \int_0^u n \kappa_z \dd z } \left (n \kappa_u \theta_u + \frac{1}{2} n(n-1) \lambda_u^2  \right ) \Ex(V_s^m V_u^{n-1} ) \dd u  \right ), \\
\Cov (V_s^m ,V_t^n) &= \Ex(V_s^m V_t^n) - \Ex(V_s^m) \Ex(V_t^n),
\end{align*}
all for $m, n \geq 1$ and $s<t$.
\end{proposition}
%-------Proof: CIR moments
\begin{proof}
We give an outline for obtaining $\Var(V_t)$ and $\Cov(V_s, V_t)$. The other terms follow a similar methodology.
Notice that $\Var(V_t) = \Ex(V_t - \Ex(V_t))^2$. Then using \cref{CIRsol0} and $\Ex(V_t)$, 
\begin{align*}
\Var(V_t) = \Ex \left ( \int_0^t e^{-\int_u^t \kappa_z \dd z } \lambda_u \sqrt{V_u} \dd B_u \right )^2 = \int_0^t e^{-2 \int_u^t \kappa_z \dd z } \lambda^2_u \Ex(V_u) \dd u.
\end{align*}
Assume $s<t$. Using the representation of $V_t$ in terms of $V_s$ \cref{CIRsol}, we have 
\begin{align*}
\Cov(V_s, V_t) &= \Cov \left ( V_s,  V_s e^{- \int_s^t \kappa_z \dd z } + \int_s^t e^{- \int_u^t \kappa_z \dd z } \kappa_u \theta_u \dd u + \int_s^t e^{- \int_u^t \kappa_z \dd z } \lambda_u \sqrt{ V_u} \dd B_u \right )\\
&=  e^{-\int_s^t \kappa_u \dd u }\Var(V_s), 
\end{align*}
where we have used that $V_s$ is independent of the It\^o integral $ \int_s^t e^{- \int_u^t \kappa_z \dd z } \lambda_u \sqrt{ V_u} \dd B_u$. 
\end{proof}

%-------Deriving the IGa moments
\subsection{Moments of the IGa process}
\label{appen:IGa}
Let $V$ be an IGa$(v_0; \kappa_t, \theta_t, \lambda_t)$. It satisfies the SDE 
\begin{align*}
\dd V_t = \kappa_t(\theta_t -  V_t) \dd t + \lambda_t  V_t \dd B_t, \quad V_0 = v_0,
\end{align*}
where we assume $(\kappa_t)_{0 \leq t \leq T}, (\theta_t)_{0 \leq t \leq T}$ and $(\lambda_t)_{0 \leq t \leq T}$ are time-dependent, deterministic, strictly positive and bounded on $[0,T]$. Then for $s<t$, $V$ has the explicit strong solution
\begin{align*}
	V_t =  V_s \frac{Y_t}{Y_s} \left ( \frac{v_0 + \int_0^t \kappa_u \theta_u/Y_u \dd u}{v_0 +\int_0^s \kappa_u \theta_u / Y_u \dd u }\right).
\end{align*}
In particular, for $s=0$, 
\begin{align*}
	V_t &= Y_t  \left ( v_0 + \int_0^t \frac{\kappa_u \theta_u}{Y_u} \dd u \right ).
\end{align*}

%-------Prop: IGa moments
\begin{proposition}
$V$ has the following moments:
\begin{align*}
\Ex( V_t^n) &= e^{\int_0^t \frac{n(n-1)}{2} \lambda_z^2 - n \kappa_z \dd z } \left (  v_0^n + n \int_0^t \kappa_u \theta_u e^{-\int_0^u \frac{n(n-1)}{2} \lambda_z^2 - n \kappa_z \dd z } \Ex(  V_u^{n-1}) \dd u \right ), \\
\Var( V_t) &= e^{-2\int_0^t \kappa_z \dd z } \int_0^t \lambda_u^2 \Ex( V_u^2) e^{2 \int_0^u \kappa_z \dd z } \dd u, \\
\Cov(  V_s,   V_t) &= \Var( V_s) e^{-\int_s^t \kappa_z \dd z }, \\ 
\Ex( V_s^m  V_t^n ) &= e^{ \int_0^t \frac{n(n-1)}{2} \lambda_z^2 - n \kappa_z \dd z } \Big ( \Ex( V_s^{m+n}) e^{-\int_0^s \frac{n(n-1)}{2} \lambda_z^2 - n \kappa_z \dd z} \\&+ n \int_s^t \kappa_u \theta_u e^{-\int_0^u \frac{n(n-1)}{2} \lambda_z^2 - n \kappa_z \dd z}  \Ex(  V_s^m   V_u^{n-1}) \dd u \Big ), \\
\Cov(  V_s^m,   V_t^n) &= \Ex(  V_s^m   V_t^n) - \Ex(  V_s^m) \Ex(  V_t^n),
\end{align*}
all for $m, n \geq 1$ and $s < t$.
\end{proposition}
%-------Proof: IGa moments
\begin{proof}
We show how to obtain $\Ex(V_s^n V_t^m)$. The other terms follow a similar methodology.
We consider the differential of $V^n$.
\begin{align*}
	\dd (V_t^n) &= \left (n \kappa_t \theta_t V_t^{n-1}+ \left (\frac{1}{2} n (n-1) \lambda_t^2 -n \kappa_t \right ) V_t^n \right ) \dd t + n\lambda_t V_t^n \dd B_t  \\
	\implies V_t^n &= V_s^n + \int_s^t n \kappa_u \theta_u V_u^{n-1} + \left (\frac{1}{2} n (n-1) \lambda_u^2  - n \kappa_u  \right ) V_u^n\dd u + \int_s^t n \lambda_u V_u^n \dd B_u. 
\end{align*}
Multiplying both sides by $V_s^m$ and taking expectation yields
\begin{align*}
	\Ex(V_s^m V_t^n) = \Ex(V_s^{n+m}) + \int_s^t n \kappa_u \theta_u \Ex(V_s^m V_u^{n-1}) +\left ( \frac{1}{2} n (n-1) \lambda_u^2 - n \kappa_u \right ) \Ex(V_s^m V_u^n) \dd u.
\end{align*}
Differentiating both sides in $t$ and letting $M^{m,n}_s(t) := \Ex(V_s^m V_t^n)$, then 
\begin{align*}
	\frac{\dd}{\dd t} M_s^{m,n}(t) = n\kappa_t \theta_t M^{m,n-1}_s(t) + \left ( \frac{1}{2}n(n-1)\lambda_t^2-n\kappa_t \right ) M^{m,n}_s(t).
\end{align*}
This is a first order ODE, which can be solved with the integrating factor method by integrating from $s$ to $t$. 
\end{proof}

%------------------------------------------------------------------------Appendix E Solutions to SDEs with linear diffusion------------------------------------
%---------------------------------------------------------------------------------------------------------------------------------------------------------------------
\section{Solutions to SDEs with linear diffusion}
\noindent
\label{appen:SDElinear}
Suppose the diffusion $U$ solves the SDE 
\begin{align}
	\dd U_t = f(t, U_t) \dd t + \nu_t U_t \dd B_t, \quad U_0 = u_0, \label{linearsde}
\end{align}
where $(\nu_t)_{0 \leq t \leq T}$ is adapted to the Brownian filtration and $f$ and $\nu$ satisfy some regularity conditions so that a pathwise unique strong solution for $U$ exists. For example, $f$ Lipschitz in $x$ uniformly in $t$, and $\nu$ bounded on $[0, T]$ is sufficient.

\begin{proposition}
The solution to \cref{linearsde} can be expressed as 
\begin{align*}
	U_t = Y_t/F_t,
\end{align*}
where $F$ is a GBM$(1; \nu_t^2, -\nu_t)$, That is, 
\begin{align*}
	\dd F_t &= \nu_t^2 F_t \dd t - \nu_t F_t \dd B_t, \quad F_0 = 1, \\
	\implies F_t &= \exp \left \{ \int_0^t \frac{1}{2} \nu_u^2 \dd u - \int_0^t \nu_u \dd B_u \right \},
\end{align*}
and $Y$ solves the integral equation (written in differential form)
\begin{align}
	\dd Y_t = F_t f \left (t, \frac{Y_t}{F_t} \right ) \dd t, \quad Y_0 = u_0 \label{ODE}.
\end{align} 
\end{proposition}
\begin{proof}
We essentially verify that this form of $U$ satisfies the SDE \cref{linearsde}.
\begin{align*}
	\dd \left ( \frac{Y_t}{F_t}\right ) &=  \dd \left ( 1/F_t \right ) Y_t + \frac{1}{F_t} \dd Y_t +  \dd \left ( 1/F_t \right )  \dd Y_t \\
	&= \left ( \left \{ \frac{\nu_t}{F_t} \dd B_t - \frac{\nu_t^2}{F_t} \dd t \right \}+ \frac{\nu_t^2}{F_t} \dd t \right ) Y_t + f \left (t, Y_t/F_t \right ) \dd t+ 0 \\
	&= \frac{Y_t}{F_t} \nu_t \dd B_t + f \left (t, Y_t/F_t \right ) \dd t\\
	&= \nu_t U_t \dd B_t + f(t, U_t) \dd t.
\end{align*}
\end{proof}

\end{document}